\documentclass[lettersize,journal]{IEEEtran}
\usepackage{cite}
\usepackage{amsmath,amssymb,amsfonts}
\usepackage{bm}
\usepackage{graphicx}
\usepackage{textcomp}
\usepackage{xcolor}
\def\BibTeX{{\rm B\kern-.05em{\sc i\kern-.025em b}\kern-.08em
    T\kern-.1667em\lower.7ex\hbox{E}\kern-.125emX}}

\usepackage{subfig}
\usepackage{tikz}
\usepackage[nolist]{acronym}
\usepackage{pgfgantt}
\usepackage[ruled,linesnumbered]{algorithm2e}
\usepackage{pgf}
\usepackage{enumitem}
\usepackage{amsthm}
\usepackage{multirow}
\usepackage{booktabs}
\usepackage{rotating}

\usetikzlibrary{decorations.pathmorphing}
\tikzset{snake it/.style={decorate, decoration={snake, amplitude=0.5mm, segment length=1mm}}}
\usetikzlibrary{arrows, positioning}
\usetikzlibrary{patterns}
\usetikzlibrary{automata, positioning}
\usepackage[capitalize]{cleveref}

\newtheoremstyle{upright}   
  {6pt}                     
  {6pt}                     
  {\normalfont}             
  {}                        
  {\bfseries \itshape}               
  {:}                       
  { }                       
  {}                        

\theoremstyle{upright}

\crefformat{algorithm}{Alg.~#2#1#3}

\newtheorem{definition}{Definition}
\crefformat{definition}{Def.~#2#1#3}
\newtheorem{listing}{Listing}
\crefformat{listing}{Lst.~#2#1#3}
\newtheorem*{example}{Example}

\newtheorem{theorem}{Theorem}

\newcounter{rownumber}[listing] 
\renewcommand{\therownumber}{\arabic{rownumber}}

\newcommand{\rowlabel}[1]{\refstepcounter{rownumber}{\footnotesize\textbf{\therownumber}}\label{#1}}

\crefname{rownumber}{Ln.}{Lns.}
\Crefname{rownumber}{Ln.}{Lns.}
\crefname{line}{Ln.}{Lns.}
\Crefname{line}{Ln.}{Lns.}
\newcommand{\revised}[1]{{#1}}

\newcommand{\NonNegativeIntegers}{\mathbb{N}_{0}}
\newcommand{\Reals}{\mathbb{R}}

\newcommand{\RationalsNonNegative}{\mathbb{Q}^+_{0}}

\newcommand{\SetCard}[1]{{|#1|}}

\newcommand{\graph}{g}

\newcommand{\InitialTokens}{\delta}

\DeclareMathAlphabet{\mathbbmsl}{U}{bbm}{m}{sl}

\newcommand{\SetActors}{A}
\newcommand{\SetChannels}{C}

\newcommand{\actor}{a}
\newcommand{\channel}{c}

\newcommand{\IdlePower}[1]{p^\mathrm{idl}(\actor_{#1})}
\newcommand{\ExecutionTime}[1]{\ell^\mathrm{exe}(\actor_{#1})}
\newcommand{\ExecutionPower}[1]{p^\mathrm{exe}(\actor_{#1})}
\newcommand{\ShutdownDelay}[1]{\ell^\mathrm{shd}(\actor_{#1})}
\newcommand{\ShutdownPower}[1]{p^\mathrm{shd}(\actor_{#1})}
\newcommand{\SleepPower}[1]{p^\mathrm{slp}(\actor_{#1})}
\newcommand{\WakeupDelay}[1]{\ell^\mathrm{wkp}(\actor_{#1})}
\newcommand{\WakeupPower}[1]{p^\mathrm{wkp}(\actor_{#1})}

\newcommand{\EnergyAlwaysActiveActor}[1]{E^{\mathrm{AA}}(P,\actor_{#1})}
\newcommand{\EnergySelfPoweredActor}[1]{E^{\mathrm{SP}}(P,\actor_{#1})}

\newcommand{\SchedStartTime}[1]{\tau_{#1}}
\newcommand{\SchedStartTimes}{\boldsymbol{\tau}}

\newcommand{\DecisionVariable}[1]{x_{#1}}
\newcommand{\DecisionVector}{\mathbf{x}}

\newcommand{\ParetoFront}{S}
\newcommand{\ParetoRef}{\ParetoFront_\mathrm{Ref}}
\newcommand{\ParetoApp}{\ParetoFront_\mathrm{App}}

\newacro{AA}{Always-Active}
\newacro{AEC}{Acoustic Echo Cancellation}
\newacro{CC}{Clock Cycle}
\newacro{CSDF}{Cyclo-Static Dataflow}
\newacro{DFG}{Dataflow Graph}
\newacro{DPM}{Dynamic Power Management}
\newacro{DSE}{Design Space Exploration}
\newacro{DVS}{Dynamic Voltage Scaling}
\newacro{HLS}{High-Level Synthesis}
\newacro{FIFO}{First In, First Out}
\newacro{ILP}{Integer Linear Program}
\newacro{IoT}{Internet of Things}
\newacro{LP}{Linear Program}
\newacro{MILP}{Mixed-Integer-Linear-Program}
\newacro{MOP}{Multi-objective Optimization Problem}
\newacro{SoC}{Systems on Chip}
\newacro{SP}{Self-Powered}
\newacro{RTL}{Register-Transfer Level}
\newacro{SDF}{Synchronous Dataflow}
\newacro{SDFG}{Synchronous Dataflow Graph}
\usepackage{glossaries}
\makeglossaries
\newacronym{hnf}{H\&S}{Hop and Skip}

\usepackage{xcolor}

\definecolor{pastel-blue}{HTML}{87CEEB}
\definecolor{pastel-red}{HTML}{FD7F6F}
\definecolor{pastel-orange}{HTML}{FFB55A}
\definecolor{pastel-yellow}{HTML}{FDED8C}
\definecolor{pastel-green}{HTML}{B2E061}
\definecolor{pastel-lavender}{HTML}{BEB9DB}

\begin{document}

\title{Exploration of Energy and Throughput Tradeoffs for Dataflow Networks

\thanks{This work has been partially funded by the Deutsche Forschungsgemeinschaft (DFG, German Research Foundation) under project grant 530178246.}

\author{\IEEEauthorblockN{Abrarul Karim, Joachim Falk, and Jürgen Teich}\\
\IEEEauthorblockA{
\textit{Friedrich-Alexander-Universität Erlangen-Nürnberg (FAU)}\\
Erlangen, Germany \\
\textit{firstname}.\textit{lastname}@fau.de}
\vspace{-5mm}
}
}

\maketitle

\begin{abstract}
The introduction of dynamic power management strategies such as clock gating and power gating in dataflow networks has been shown to provide significant energy savings when applied during idle times.
However, these strategies can also degrade throughput due to shutdown and wake-up delays.
Such throughput degradations might be particularly detrimental to signal processing systems that require a guaranteed throughput.
\par
As a solution, this paper first contributes a linear-program formulation for finding a periodic maximal-throughput schedule of a given so-called \emph{self-powering dataflow network} where actors, realized in hardware, are allowed to go to sleep whenever not being enabled to fire.
Depending on which actors are allowed to power down, tradeoffs between throughput and energy savings can be obtained.
As a second contribution, we propose a mixed-integer-linear-program formulation to determine a periodic schedule that satisfies a given throughput while minimizing the overall energy per period by identifying a respective set of actors that is allowed to power down in phases of idleness and which rather not.
Finally, as a third contribution, we propose a multi-objective design-space exploration strategy called ``Hop and Skip'' to efficiently explore the Pareto front of energy and throughput solutions.
\par
Experimental evaluations on a set of existing benchmarks and randomly generated graphs witness significant exploration time reductions over a brute-force sweep.
Finally, a real-world case study is elaborated and we report on achievable energy savings and throughputs of the related dataflow network where (a) all actors are always-active, (b) all actors are self-powered, and (c) all optimal energy and throughput tradeoff points as found by the proposed design-space exploration strategy.
\end{abstract}


\section{Introduction}\label{sec:introduction}

Modern signal processing applications, such as image and video analysis, face the challenge of being both computationally demanding and energy-constrained, arising from their realization in \ac{IoT} devices that rely on ambient energy sources or batteries~\cite{Wei18EdgeComputingSurvey,Jolly21IOTBatteryLife}.
To model these applications, we can use \acp{DFG}, which consist of a network of interconnected \textit{actors}—the computational kernel nodes—and \textit{channels}—the communication edges~\cite{FalkHZT13,FalkZHT13}.
An actor is enabled to execute once specific conditions, referred to as \textit{firing rules}, are met.
Upon firing, \emph{tokens} are consumed from the input channels, and resultant tokens are produced on the output channels of an actor.
\acp{DFG} enable the exploitation of concurrency of multiple actors operating and firing concurrently (global parallelism) as well as within the implementation of actors (local parallelism).
To exploit such parallelism, there exist code-generation frameworks such as~\cite{ReicheSHMT14} that transform applications concisely described in a domain-specific language into hardware implementations that realize each application as a \ac{DFG}.
\par
Recently, there has been a growing focus in research on developing implementations of dataflow applications that are not only throughput-optimized but also energy-efficient~\cite{Lerner24MordernHWDataProcessing}.
The concept of \textit{Self-powering Dataflow Networks} \cite{karim2024-selfpoweringDFGs,karim2025-selfpoweringDFGs-mbmv} exploits the idea to allow actor implementations to power down during idle times, for example, when waiting for input data.
However, bringing an actor into sleep mode and waking it up again may introduce not only the problem of conditional execution times, but even non-negligible additional execution-time delays, raising the two issues of (a) how to properly analyze systems that exhibit dynamically arising execution-time delays, and (b) deciding if an actor should shut down at all to \revised{meet} a given throughput constraint.
\par
The paper is structured as follows: \Cref{sec:fundamentals} introduces the required fundamentals of \acp{DFG}, self-powered dataflow, and throughput analysis.
The 1st issue is addressed in \cref{sec:critactors} by proposition of a \ac{LP} formulation for analyzing the throughput of self-powering dataflow networks, while the 2nd issue is solved by introducing a \ac{MILP} formulation to determine a periodic schedule that satisfies a given throughput constraint while minimizing the overall energy per period by identifying a respective set of actors that is allowed to power down in phases of idleness and which rather not.
\Cref{sec:dse} presents our 3rd contribution, a novel \ac{DSE} strategy called \textit{\gls{hnf}}, to determine tradeoffs between throughput and energy across different \ac{DFG} implementations.
\Cref{sec:results} introduces a real-world case study to analyze achievable energy savings and throughputs of \acp{DFG} where (a) all actors are always-active, (b) all actors are self-powered, and (c) all optimal energy and throughput tradeoff points, as found by the proposed \ac{DSE} strategy.
Additionally, evaluations on a set of \ac{DFG} benchmarks and $100$ random \acp{DFG} demonstrate that the \gls{hnf} strategy significantly reduces exploration time.
Finally, related work is discussed in~\cref{sec:relwork}, and \cref{sec:conclusion} concludes the paper.

\section{Fundamentals and Background}\label{sec:fundamentals}

In this section, the fundamentals of \acfp{DFG} and the necessary notations for scheduling, power modeling, and analysis of hardware implementations of \acp{DFG} are discussed.
Also, the concept of self-powering dataflow networks~\cite{karim2024-selfpoweringDFGs} is introduced.

\subsection{Dataflow Graphs, Scheduling, and Power Modeling}

In \acp{DFG}, \emph{actors} model functionality, whereas \emph{channels} between actors model data communication and storage.
Based on actors and channels, a graph is formed that is defined as follows:

\begin{definition}[Dataflow Graph~\cite{E.Lee87SDF}]\label{def:dfg}
  A \ac{DFG} is a directed graph $\graph = (\SetActors, \SetChannels)$ containing a set of actors $\SetActors$ and a set of directed edges $\SetChannels \subseteq \SetActors \times \SetActors$ representing \ac{FIFO} channels, each associated with a number of initial tokens given by a function $\InitialTokens: \SetChannels \rightarrow \NonNegativeIntegers$.
\end{definition}

In this paper, we consider \acp{DFG} in which an actor is enabled to fire if there exists at least one token on each of its incoming channels.
When firing, it consumes exactly one token on each input and also produces exactly one token at each outgoing channel.
Such graphs are called \emph{marked graphs}~\cite{CommonerHolt71MarkedGraphs}.

\begin{example}
  \Cref{fig:firing} illustrates the scenario of an actor $\actor_1 \in \SetActors$ being enabled for firing in (b), and also the corresponding scenario after actor $\actor_1$ has fired in (c).
  Note that an actor $\actor_i$, once enabled for firing, does not necessarily need to fire immediately.
\end{example}
  
\begin{figure}[htb]
  \vspace{-5mm}
  \centering\scalebox{0.87}{\subfloat[Not fireable \label{fig:firing-nonfireable}]{
    \begin{tikzpicture}[->, node distance=1.5cm,>=stealth', scale=1, every node/.style={scale=1}]
        \node[rectangle, draw, rounded corners, minimum width=1cm, minimum height=0.75cm] (A1) {$\actor_1$};  
        
        \node[left of=A1, yshift=0.5cm] (I1) {};
        \node[left of=A1, xshift=0cm] (I2) {};
        \node[left of=A1, yshift=-0.5cm] (I3) {};
    
        \node[right of=A1, yshift=-0.35cm] (O1) {};
        \node[right of=A1, yshift=0.35cm] (O2) {};
    
        \draw[thin] (I1) -- node {} (A1);
        \draw[thin] (I2) -- node[pos=0.5, fill=black, circle, inner sep=0.5pt, minimum size=5pt] {} (A1);
        \draw[->, thin] (I3) to
        node[pos=0.3, fill=black, circle, inner sep=0.5pt, minimum size=5pt] {}
        node[pos=0.6, fill=black, circle, inner sep=0.5pt, minimum size=5pt] {} (A1);
    
        \draw[thin] (A1) to (O1);
        \draw[thin] (A1) to (O2);
    \end{tikzpicture}
}
\subfloat[Fireable \label{fig:firing-fireable}]{
    \begin{tikzpicture}[->, node distance=1.5cm,>=stealth', scale=1, every node/.style={scale=1}]
        \node[rectangle, draw, rounded corners, minimum width=1cm, minimum height=0.75cm] (A1) {$\actor_1$};  
        
        \node[left of=A1, yshift=0.5cm] (I1) {};
        \node[left of=A1, xshift=0cm] (I2) {};
        \node[left of=A1, yshift=-0.5cm] (I3) {};
    
        \node[right of=A1, yshift=-0.35cm] (O1) {};
        \node[right of=A1, yshift=0.35cm] (O2) {};
    
        \draw[thin] (I1) -- node[pos=0.5, fill=black, circle, inner sep=0.5pt, minimum size=5pt] {} (A1);
        \draw[thin] (I2) -- node[pos=0.5, fill=black, circle, inner sep=0.5pt, minimum size=5pt] {} (A1);
        \draw[->, thin] (I3) to
        node[pos=0.3, fill=black, circle, inner sep=0.5pt, minimum size=5pt] {}
        node[pos=0.6, fill=black, circle, inner sep=0.5pt, minimum size=5pt] {} (A1);    
        \draw[thin] (A1) to (O1);
        \draw[thin] (A1) to (O2);
    \end{tikzpicture}
}
\subfloat[After firing]{
    \begin{tikzpicture}[->, node distance=1.5cm,>=stealth', scale=1, every node/.style={scale=1}]
        \node[rectangle, draw, rounded corners, minimum width=1cm, minimum height=0.75cm] (A1) {$\actor_1$};    
        \node[left of=A1, yshift=0.5cm] (I1) {};
        \node[left of=A1, xshift=0cm] (I2) {};
        \node[left of=A1, yshift=-0.5cm] (I3) {};
    
        \node[right of=A1, yshift=-0.35cm] (O1) {};
        \node[right of=A1, yshift=0.35cm] (O2) {};
    
        \draw[thin] (I1) -- node {} (A1);
        \draw[thin] (I2) -- node {} (A1);
        \draw[thin] (I3) -- node[pos=0.5, fill=black, circle, inner sep=0.5pt, minimum size=5pt] {} (A1);
    
        \draw[thin] (A1) -- node[pos=0.5, fill=black, circle, inner sep=0.5pt, minimum size=5pt] {}(O1);
        \draw[thin] (A1) -- node[pos=0.5, fill=black, circle, inner sep=0.5pt, minimum size=5pt] {} (O2);
    \end{tikzpicture}
}}
  \caption{\label{fig:firing}
    The shown actor $\actor_1$ has three input channels (incoming edges) and two output channels (outgoing edges), while the black dots on these edges represent tokens.
    In (a), $\actor_1$ is \emph{not fireable} due to the absence of a token on its upper input channel.
    In (b), $\actor_1$ is \emph{fireable} because at least one token is available on each of its input channels.
    Finally, (c) depicts the situation after $\actor_1$ has fired by consuming one token from each input channel and producing one token on each output channel.}
\end{figure}
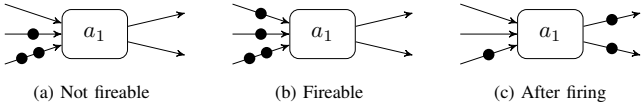

In the following, we are interested in the existence of \emph{periodic schedules} and the question of the maximal achievable throughput for a given marked graph $\graph$ by a periodic schedule $\SchedStartTimes$, defined as follows:
\par
\begin{definition}[Periodic Schedule]\label{def:periodic-schedule}
A periodic schedule $\SchedStartTimes$ of period $P$ for a \ac{DFG} $\graph = (\SetActors,\SetChannels)$ is a vector $\SchedStartTimes = (\SchedStartTime{1}, \SchedStartTime{2}, \ldots \SchedStartTime{\SetCard{\SetActors}})$ that assigns each actor $\actor_i \in \SetActors$ a start time $\SchedStartTime{i}$
such that
\begin{equation}
  \SchedStartTime{j} + \InitialTokens(\actor_i, \actor_j) \cdot P \ge \SchedStartTime{i} + \ExecutionTime{i} \quad \forall (\actor_i, \actor_j) \in \SetChannels,\label{eq:data-dependencies}
\end{equation}
where $\ExecutionTime{i} \in \NonNegativeIntegers$ corresponds to the \emph{execution time} of actor $\actor_i$ that includes the time to read input data, perform an actor-specified function of the input data, and finally consume the input and produce output data, while $\InitialTokens(\actor_i, \actor_j)$ denotes the number of initial tokens on the channel $(\actor_i, \actor_j)$ between actor $\actor_i$ and actor $\actor_j$.
In such a periodic schedule, the $n$-th firing of the actor $\actor_i$ starts at time $\SchedStartTime{i} + (n-1)\cdot P$.
Thus, $\SchedStartTime{i}$ denotes the start time of the first firing $(n=1)$ of each actor $\actor_{i} \in \SetActors$.
\end{definition}
\par
Regarding the question of the \emph{maximal achievable throughput}, i.e., corresponding to the \emph{minimal period} $P_\mathrm{min}$, it can be shown that for marked graphs, a \emph{period lower bound} $P_\mathrm{min}$ can be determined by the so-called \emph{maximal cycle mean}~\cite{Fet76,Parhi91OptUnfold} as follows:
\par
\begin{definition}[Maximum Cycle Mean~\cite{Parhi91OptUnfold}]
  For a \ac{DFG} $\graph = (\SetActors,\SetChannels)$ with at least one \emph{directed cycle}, the period lower bound $P_\mathrm{min}$ is given by
  \vspace{-4mm}
  \begin{equation} \label{eqn:max-cycle-mean}
    P_\mathrm{min} \;=\; \max_{z \in Z} \frac
        {\sum\limits_{\actor_i \in z \cap \SetActors}     \ExecutionTime{i}}
        {\sum\limits_{\channel_k \in z \cap \SetChannels} \InitialTokens(\channel_k)},
  \end{equation}
  where $Z$ denotes the set of directed cycles of graph $\graph$.%
  \footnote{If we allocate just one resource instance for each actor $\actor_i$ (as in the case of a direct hardware implementation of a \ac{DFG}), then $P_{\mathrm min}$ is additionally bounded by all self-cycles $\channel = (\actor_i, \actor_i)$ with weight $\ExecutionTime{i}$ and $\InitialTokens(\channel) = 1$.}
\end{definition}
\par
A \textit{self-timed} schedule in which each actor fires immediately when becoming enabled to fire can be shown to always realize this period lower bound~\cite{moreira2007self}.
\par
\begin{example}
As our running example, a \ac{DFG} of an \ac{AEC} network architecture is illustrated in~\cref{fig:echo-cancellation}.
Shorthand names are provided for each actor and are used interchangeably.
The actor \texttt{src} provides an input signal, which we call the playback signal, to the network, which is played back on a speaker modelled by the actor \texttt{spk}.
The actor \texttt{mic} captures a second microphone input signal, which we call the recorded signal, which includes the echo of the playback signal.
The echo needs to be mitigated before being sent to the actor \texttt{snk}.
The recorded signal is initially passed into an anti-aliasing/low-pass filter implemented in actor \texttt{filt}.
The actor \texttt{sum} subtracts the predicted echo signal (computed by the actor \texttt{conv}) from the filtered signal.
The mitigated signal from \texttt{sum} is then passed onto actor \texttt{dup2}, which duplicates and forwards it to actor \texttt{adap}, as well as actor \texttt{sink} for further use.
The actor \texttt{dup3} replicates the playback signal and sends it to the actors \texttt{adap}, \texttt{conv}, and \texttt{spk}.
Convolution filter coefficients are adapted by the actor \texttt{adap} based on the mitigated signal as well as the playback signal, and forwarded to  actor \texttt{conv}, which generates the predicted echo signal.

\begin{figure}[ht]
  \vspace{-1mm}
  \centering \scalebox{0.95}{\begin{tikzpicture}[->, node distance=2.5cm, >=stealth', scale=1, every node/.style={scale=1}]

    \node[draw, rectangle, rounded corners, minimum width=1.5cm, minimum height=0.5cm] (mic) {$\actor_{mic} \mid \texttt{mic}$};
    \node[draw,fill=lightgray, rectangle, rounded corners, minimum width=1.5cm, minimum height=0.5cm, below of=mic, yshift=1cm] (filt) {$\actor_1 \mid \texttt{filt}$};
    \node[draw,fill=lightgray, rectangle, rounded corners, minimum width=1.5cm, minimum height=0.5cm, right of=filt] (sum) {$\actor_2 \mid \texttt{sum}$};
    \node[draw,fill=lightgray, rectangle, rounded corners, minimum width=1.5cm, minimum height=0.5cm, right of=sum] (dup2) {$\actor_3 \mid \texttt{dup2}$};
    \node[draw,fill=lightgray, rectangle, rounded corners, minimum width=1.5cm, minimum height=0.5cm, above of=dup2,yshift=-1cm] (adap) {$\actor_4 \mid \texttt{adap}$};
   \node[draw,fill=lightgray, rectangle, rounded corners, minimum width=1.5cm, minimum height=0.5cm, above of=sum,yshift=-1cm] (conv) {$\actor_5 \mid \texttt{conv}$};
   \node[draw, fill=lightgray, rectangle, rounded corners, minimum width=1.5cm, minimum height=0.5cm, above of=conv,yshift=-1cm,xshift=1.25cm] (dup_3) {$\actor_6 \mid \texttt{dup3}$};
   \node[draw, rectangle, rounded corners, minimum width=1.5cm, minimum height=0.5cm, right of=dup2] (sink) {$\actor_{snk} \mid \texttt{sink}$};\textbf{}
   \node[draw, rectangle, rounded corners, minimum width=1.5cm, minimum height=0.5cm, above of=mic, yshift=-1cm] (spk) {$\actor_{spk} \mid \texttt{spk}$};
   \node[draw, rectangle, rounded corners, minimum width=1.5cm, minimum height=0.5cm, right of=adap, yshift=1.5cm] (src) {$\actor_{src} \mid \texttt{src}$};

    \draw (conv) -- node[pos=0.4, fill=black, circle, inner sep=0.5pt, minimum size=5pt] {} (sum);
    \draw (mic) -- (filt);
    \draw (filt) -- (sum);
    \draw (sum) -- (dup2);
    \draw (dup2) -- (sink);
    \draw (dup2) -- (adap);
    \draw (adap) -- (conv);
    \draw (dup_3) -- (conv);
    \draw (dup_3) -- (spk);
    \draw (dup_3) -- (adap);
    \draw (src) -- (dup_3);

\end{tikzpicture}}
  \vspace{-4mm}
  \caption{\label{fig:echo-cancellation} \ac{DFG} of an \ac{AEC} application, detailing the flow of data through actors (nodes) interconnected by \ac{FIFO} channels (directed edges).
  A single initial token is present on the channel $(\actor_5, \actor_2)$, i.e., $\InitialTokens(\actor_5, \actor_2) = 1$.}
  \vspace{-2mm}
\end{figure}
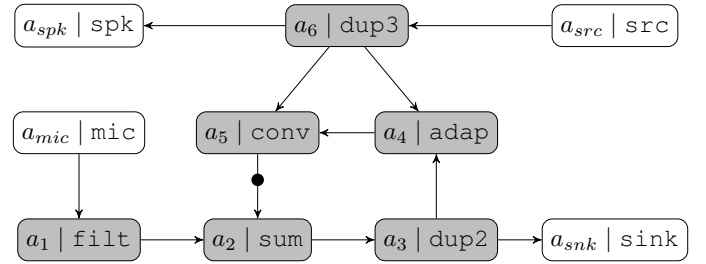

The white actors of the \ac{AEC} application are input and output actors, and not considered part of a \ac{DFG} implementation, while the shaded actors of the \ac{AEC} application need to be implemented, and hence, form the set of actors subject to analysis, i.e., $\SetActors = \{ \actor_1, \actor_2, \ldots\actor_6 \}$.
For the actor execution times, we assume $\ExecutionTime{1}=\ExecutionTime{6}=4, \ExecutionTime{2}=\ExecutionTime{3}=3, \ExecutionTime{4}=9$, and $\ExecutionTime{5}=8$ time units.
\par
\cref{eqn:max-cycle-mean} delivers $P_\mathrm{min} = 23$ as the \ac{DFG} of the \ac{AEC} application contains exactly one directed cycle $z_1 = \{ \actor_2,$ $(\actor_2, \actor_3),$ $\actor_3,$ $(\actor_3, \actor_4),$ $\actor_4,$ $(\actor_4, \actor_5),$ $\actor_5,$ $(\actor_5, \actor_2) \}$, $\ExecutionTime{2} + \ExecutionTime{3} + \ExecutionTime{4} + \ExecutionTime{5} = 23$, and $\InitialTokens(\actor_2, \actor_3) + \InitialTokens(\actor_3, \actor_4) + \InitialTokens(\actor_4, \actor_5) +\InitialTokens(\actor_5, \actor_2) = 1$.
Moreover, as shown in~\cref{fig:anc-schedule}, there exists a periodic schedule $\SchedStartTimes = (\SchedStartTime{1}, \SchedStartTime{2}, \ldots\SchedStartTime{6}) = (0,4,7,10,19,6)$ realizing the period $P_\mathrm{min}$.

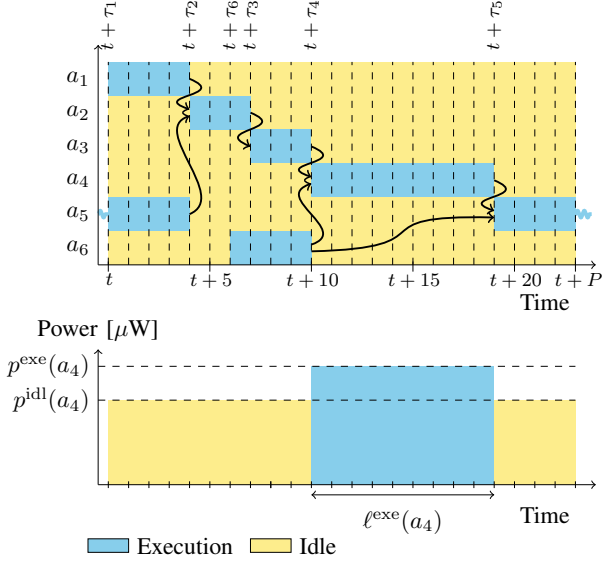
\begin{figure}[ht]
  \vspace{-3mm}
  \centering
  \scalebox{0.9}{\begin{tikzpicture}[xscale=0.3, yscale=0.5]


    

    

    \node[left] at (0.5,0.5) {$\actor_6$};
    \node[left] at (0.5,1.5) {$\actor_5$};
    \node[left] at (0.5,2.5) {$\actor_4$};
    \node[left] at (0.5,3.5) {$\actor_3$};
    \node[left] at (0.5,4.5) {$\actor_2$};
    \node[left] at (0.5,5.5) {$\actor_1$};

    

    \fill[pastel-yellow] (0 +1,0) rectangle (23 +1,6); 
    \fill[pastel-blue]  ( 0 +1,5) rectangle (4 +1,6);
    \fill[pastel-blue]  ( 4 +1,4) rectangle (7 +1,5);
    \fill[pastel-blue]  ( 7 +1,3) rectangle (10 +1,4);
    \fill[pastel-blue]  (10 +1,2) rectangle (19 +1,3);
    \fill[pastel-blue]  (19 +1,1) rectangle (23 +1,2);
    \fill[pastel-blue]  ( 0 +1,1) rectangle (4 +1,2);
    \fill[pastel-blue]  ( 6 +1,0) rectangle (10 +1,1); 
    

    \draw[dashed] (1,0) -- (1,6);
    \draw[dashed] (2,0) -- (2,6);
    \draw[dashed] (3,0) -- (3,6);
    \draw[dashed] (4,0) -- (4,6);
    \draw[dashed] (5,0) -- (5,6);
    \draw[dashed] (6,0) -- (6,6);
    \draw[dashed] (7,0) -- (7,6);
    \draw[dashed] (8,0) -- (8,6);
    \draw[dashed] (9,0) -- (9,6);
    \draw[dashed] (10,0) -- (10,6);
    \draw[dashed] (11,0) -- (11,6);
    \draw[dashed] (12,0) -- (12,6);
    \draw[dashed] (13,0) -- (13,6);
    \draw[dashed] (14,0) -- (14,6);
    \draw[dashed] (15,0) -- (15,6);
    \draw[dashed] (16,0) -- (16,6);
    \draw[dashed] (17,0) -- (17,6);
    \draw[dashed] (18,0) -- (18,6);
    \draw[dashed] (19,0) -- (19,6);
    \draw[dashed] (20,0) -- (20,6);
    \draw[dashed] (21,0) -- (21,6);
    \draw[dashed] (22,0) -- (22,6);
    \draw[dashed] (23,0) -- (23,6);
    \draw[dashed] (24,0) -- (24,6);


    \node[below, rotate=90] at  (1-0.75,7) {\footnotesize{ $t+\SchedStartTime{1}$}};
    \node[below, rotate=90] at  (5-0.75,7) {\footnotesize{ $t+\SchedStartTime{2}$}};
    \node[below, rotate=90] at  (8-0.75,7) {\footnotesize{ $t+\SchedStartTime{3}$}};
    \node[below, rotate=90] at  (7-0.75,7) {\footnotesize{ $t+\SchedStartTime{6}$}};
    \node[below, rotate=90] at (11-0.75,7) {\footnotesize{ $t+\SchedStartTime{4}$}};
    \node[below, rotate=90] at (20-0.75,7) {\footnotesize{ $t+\SchedStartTime{5}$}};

    \draw[-]  (1,-0.2) -- (1,0.1);
    \node[below] at (1-0.1, 0) {\footnotesize{ $t$}};

    \draw[-]  (6,-0.2) -- (6,0.1);
    \node[below] at (6-0.1, 0) {\footnotesize{ $t+5$}};

    \draw[-]  (11,-0.2) -- (11,0.1);
    \node[below] at (11-0.1, 0) {\footnotesize{ $t+10$}};

    \draw[-]  (16,-0.2) -- (16,0.1);
    \node[below] at (16-0.1, 0) {\footnotesize{ $t+15$}};

    \draw[-]  (21,-0.2) -- (21,0.1);
    \node[below] at (21-0.1, 0) {\footnotesize{ $t+20$}};

    \draw[-]  (24,-0.2) -- (24,0.1);
    \node[below] at (24,0) {\footnotesize{ $t+P$}};

    \draw[->, thick] (5,5.5) .. controls (7,5) and (3,5) .. (5,4.6);
    \draw[->, thick] (5,1.5) .. controls (7,2) and (3,4) .. (5,4.4);
    \draw[->, thick] (8,4.5) .. controls (10,4) and (6,4) .. (8,3.5);
    \draw[->, thick] (11,3.5) .. controls (13,3) and (9,3) .. (11,2.6);
    \draw[->, thick] (11,0.6) .. controls (13,1) and (9,2) .. (11,2.4);
    \draw[->, thick] (20,2.5) .. controls (22,2) and (18,2) .. (20,1.6);
    \draw[->, thick] (11,0.4) .. controls (18,0.5) and (13,1.5)  .. (20,1.4);
    
    \tikzset{snake it/.style={decorate, decoration={snake, amplitude=0.5mm, segment length=1mm}}}
    
    \usetikzlibrary{decorations.pathmorphing}
    \draw[snake it, ultra thick, color=pastel-blue] (24,1.5) -- (24.8,1.5);
    \draw[snake it, ultra thick, color=pastel-blue] (0.5,1.5) -- (1,1.5);

    \draw[->] (0.5,0) -- (25,0); 
    \draw[->] (0.5,0) -- (0.5,6.5);
    \node[below] at (22.5,-0.6) {Time};

\begin{scope}[shift={(0,-6.5)}, every node/.style={anchor=west}]

    \fill[pastel-yellow] (1,0) rectangle (11,2.5);
    \fill[pastel-yellow] (20,0) rectangle (24,2.5);
    \fill[pastel-blue] (11,0) rectangle (20,3.5);

    \draw[<->] (11,-0.3) -- (20,-0.3);
    \node[below] at (15.5,-0.5) {$\ExecutionTime{4}$};

    
    \draw[dashed] (0.5,2.5) -- (24,2.5);
    \draw[dashed] (0.5,3.5) -- (24,3.5);

    \node[left] at (0.5,2.5) {$\IdlePower{4}$};
    \node[left] at (0.5,3.5) {$\ExecutionPower{4}$};

    \draw[-] (9,-0.1) -- (9,0.1);
    \draw[-] (1,-0.1) -- (1,0.1);
    \draw[-] (2,-0.1) -- (2,0.1);
    \draw[-] (3,-0.1) -- (3,0.1);
    \draw[-] (4,-0.1) -- (4,0.1);
    \draw[-] (5,-0.1) -- (5,0.1);
    \draw[-] (6,-0.1) -- (6,0.1);
    \draw[-] (7,-0.1) -- (7,0.1);
    \draw[-] (8,-0.1) -- (8,0.1);
    \draw[-] (9,-0.1) -- (9,0.1);
    \draw[-] (10,-0.1) -- (10,0.1);
    \draw[-] (11,-0.1) -- (11,0.1);
    \draw[-] (12,-0.1) -- (12,0.1);
    \draw[-] (13,-0.1) -- (13,0.1);
    \draw[-] (14,-0.1) -- (14,0.1);
    \draw[-] (15,-0.1) -- (15,0.1);
    \draw[-] (16,-0.1) -- (16,0.1);
    \draw[-] (17,-0.1) -- (17,0.1);
    \draw[-] (18,-0.1) -- (18,0.1);
    \draw[-] (19,-0.1) -- (19,0.1);
    \draw[-] (20,-0.1) -- (20,0.1);
    \draw[-] (21,-0.1) -- (21,0.1);
    \draw[-] (22,-0.1) -- (22,0.1);
    \draw[-] (23,-0.1) -- (23,0.1);
    \draw[-] (24,-0.1) -- (24,0.1);

    \draw[->] (0.5,0) -- (0.5,4) node[above] {Power [$\mu\text{W}$]};
    \draw[->] (0.5,0) -- (25,0); 
    \node[below] at (22.5,-0.4) {Time}; 
\end{scope}

  \begin{scope}[shift={(0,-8.5)}, every node/.style={anchor=west}]
        \draw[fill=pastel-blue, draw=black] (0,0) rectangle (2,0.4);
        \node at (2,0.2) {Execution};
        
        \draw[fill=pastel-yellow, draw=black] (8,0) rectangle (10,0.4);
        \node at (10,0.2) {Idle};
        
        
        
    \end{scope}
    
\end{tikzpicture}}
  \vspace{-1mm}
  \caption{\label{fig:anc-schedule}%
    Schedule of the \ac{AEC} \ac{DFG} (see~\cref{fig:echo-cancellation}) with minimal period $P = P_{\mathrm{min}} = 23$ (top).
    The arrows represent data dependencies.
    Shown as well is the power profile of actor $\actor_4$ when operating in always-active mode (bottom).
    During the execution phase of the actor (blue) a higher power $\ExecutionPower{4}$ is consumed than the power $\IdlePower{4}$ in the idle phase (yellow).} 
  \vspace{-2mm}
\end{figure}
\end{example}

Now, a \ac{DFG} could be implemented one-to-one in hardware by generating a dedicated actor circuit for each actor $\actor_i \in \SetActors$ and implementing each channel as a \ac{FIFO} buffer, connecting the two communicating actors directly.
An obvious benefit of implementing a data-driven activation of each actor $\actor_i$ in hardware is that in times of idleness in a given schedule, an actor could go into a low-power state.

\begin{example}
For our running example, the lower plot of \cref{fig:anc-schedule} shows such a power profile of actor $\actor_4$.
In a periodic actor schedule with period $P$, only in the time interval from $\SchedStartTime{i}$ to $\SchedStartTime{i} + \ExecutionTime{i}$ is the high power consumption of the \emph{actor execution power} $\ExecutionPower{i}$ experienced.
In the remaining time interval of length $P - \ExecutionTime{i}$ within the period of length $P$, a lower power consumption $\IdlePower{i}$ called \emph{actor idle power} is asserted.
In the following, we refer to a mode of execution where an actor is able to execute immediately upon fireability as the \emph{\acf{AA}} mode.

With these model assumptions, we can estimate the energy consumption per iteration of an always-active actor $\actor_i$ scheduled within a given period $P$ as follows:
\begin{equation}\label{eq:E_Always_On}
  \begin{split}
    \EnergyAlwaysActiveActor{i}
      &= \ExecutionPower{i} \cdot \ExecutionTime{i} \\
      &+ \IdlePower{i} \cdot (P - \ExecutionTime{i})
  \end{split}
\end{equation}
\end{example}

\subsection{Self-Powering Dataflow Networks}
In \cite{karim2024-selfpoweringDFGs}, the concept of \emph{self-powering dataflow networks} has been introduced.
Motivated by the fact that in recent circuit technologies, even the idle power of actors can be quite substantial, the authors of \cite{karim2025-selfpoweringDFGs-mbmv} proposed to fully exploit modern power-saving techniques like \emph{clock-gating} and \emph{power-gating} to reduce the power and thereby energy consumption of circuits implementing DFGs further. 
\par
A mode where, during times of non-fireability, an actor is put into a \emph{low-power (sleep) mode} using techniques such as clock-gating or power-gating will be referred to as the \emph{\acf{SP}} mode of operation in the following.
\begin{example}
Continuing our running example, each mode can be applied to each actor individually.
For example,~\cref{fig:profile-sleep} shows actor $\actor_4$ executing in self-powered mode.
\par
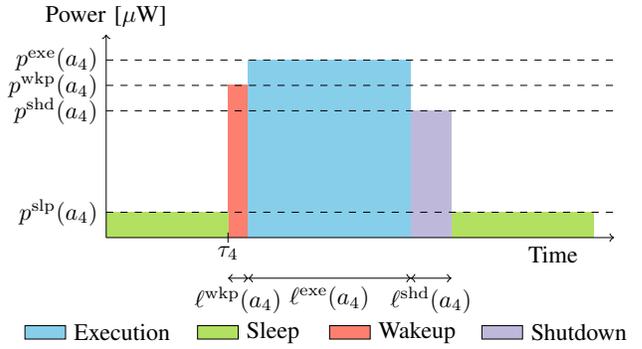
\begin{figure}[ht]
  \centering
  \scalebox{0.9}{{\begin{tikzpicture}[xscale=0.3, yscale=0.75]


    \def\xaStart{0}
    \def\xaEnd{25}
    \def\yaStart{6}
    \def\yaEnd{7.5}

    \def\xbStart{0}
    \def\xbEnd{25}
    \def\ybStart{0}
    \def\ybEnd{4}
    \fill[pastel-red] (6,0) rectangle (7,3.0);
    \fill[pastel-blue] (7,0) rectangle (15,3.5);
    \fill[pastel-lavender] (15,0) rectangle (17,2.5);
    \fill[pastel-green] (0,\ybStart) rectangle (6,\ybStart+0.5);
    \fill[pastel-green] (17,\ybStart) rectangle (24,\ybStart+0.5);

    \draw[->] (\xbStart,\ybStart) -- (\xbEnd,0);
    \node[below] at (22,0) {Time}; 
    \draw[->] (\xbStart,\ybStart) -- (\xaStart,\ybEnd) node[above] {Power [$\mu\text{W}$]};

    \draw[<->] (6,-0.8) -- (7,-0.8); 
    \node[below] at (6.5,-0.8) {$\WakeupDelay{4}$};
    \draw[<->] (7,-0.8) -- (15,-0.8); 
    \node[below] at (11,-0.8) {$\ExecutionTime{4}$};
    \draw[<->] (15,-0.8) -- (17,-0.8); 
    \node[below] at (16,-0.8) {$\ShutdownDelay{4}$};

    
    \node[below] at (6,0) {$\SchedStartTime{4}$};
    \draw[-] (6,-0.1) -- (6,0.1);

    \node[left] at (0,3.5) {$\ExecutionPower{4}$};
    \node[left] at (0,3.0) {$\WakeupPower{4}$};
    \node[left] at (0,2.5) {$\ShutdownPower{4}$};
    \node[left] at (0,0.5) {$\SleepPower{4}$};

    \draw[dashed] (0,0.5) -- (25,0.5);
    \draw[dashed] (0,3.5) -- (25,3.5);
    \draw[dashed] (0,2.5) -- (25,2.5);
    \draw[dashed] (0,3.0) -- (25,3.0);

    \begin{scope}[shift={(-4,-2)}, every node/.style={anchor=west}]
        \draw[fill=pastel-blue, draw=black] (0,0) rectangle (2,0.264);
        \node at (2,0.13) {Execution};
        
        \draw[fill=pastel-green, draw=black] (8.5,0) rectangle (10.5,0.264);
        \node at (10.5,0.13) {Sleep};
        
        \draw[fill=pastel-red, draw=black] (15,0) rectangle (17,0.264);
        \node at (17,0.13) {Wakeup};
        
        \draw[fill=pastel-lavender, draw=black] (22.5,0) rectangle (24.5,0.264);
        \node at (24.5,0.13) {Shutdown};
        
        
    \end{scope}

\end{tikzpicture}}}
  \caption{\label{fig:profile-sleep}%
    Power consumption and execution phases of actor $\actor_4$ of the running example when operating in self-powered mode.}
\end{figure}

In the above figure, the power profile of actor $\actor_4$ is shown, assuming that it shuts down into a low-power (sleep) mode immediately after detecting insufficient data availability at its inputs for committing a next firing at the time of finishing its execution (firing).
In the shown schedule, fireability is established at time $\SchedStartTime{4}$.
However, the actor now cannot execute immediately, but must first experience a wake-up delay (latency) $\WakeupDelay{4}$.
Similarly, it usually takes some non-zero time $\ShutdownDelay{4}$ to shut down an actor. 
Note that wake-up and shutdown not only induce delays but also incur power consumption overheads, as modeled by power values $\WakeupPower{4}$ and  $\ShutdownPower{4}$, respectively.
However, as a benefit, the actor will be \emph{dormant} (or in sleep mode) for a duration of $P - \WakeupDelay{4} - \ExecutionTime{4} - \ShutdownDelay{4}$, during which it consumes only $\SleepPower{4}$, which is typically much less than the idle power $\IdlePower{4}$ consumed if the actor would operate in always-active mode.
\end{example}
\par
Using these model assumptions, the energy consumption per iteration of a self-powered actor $\actor_i$ scheduled with a period $P$ can be modeled as follows:
\begin{equation}\label{eq:E_Self_Powered}
\begin{split}
&\EnergySelfPoweredActor{i}=\WakeupPower{i} \cdot \WakeupDelay{i}\\
&\quad+\ExecutionPower{i} \cdot \ExecutionTime{i}\\
&\quad+\ShutdownPower{i} \cdot \ShutdownDelay{i}\\
&\quad+\SleepPower{i} \cdot (P - \WakeupDelay{i} - \ExecutionTime{i} - \ShutdownDelay{i})
\end{split}
\end{equation}
\par
However, if an actor operates in self-powered mode, \revised{a} throughput penalty may result due to the induced wake-up delay.
Our techniques presented in the following shall analyze throughput and energy-saving tradeoffs by choosing whether to implement an actor to execute in always-active or in self-powered mode, calling such implementations \emph{hybrid}.
But first, we want to reason about the existence of periodic schedules for such hybrid implementations at all, as execution delays obviously depend on whether an actor is executed in always-active or in self-powered mode.

\section{Maximal Throughput and Minimal Energy Schedules for Hybrid \acs{DFG} Implementations}\label{sec:critactors}

First, we will prove the existence of periodic schedules for hybrid \ac{DFG} implementations by proposing a \acf{LP} formulation in~\cref{sec:lp}.
Then, assuming we want to sustain a given period $P$, \cref{sec:milp} shows how to determine a periodic schedule with period $P$ with minimal energy consumption by solving a \acf{MILP} that identifies a set of so-called \emph{critical actors}.
Each actor identified as critical must execute in always-active mode, as otherwise, the induced wake-up or sleep delays would no longer allow for a schedule with the given period $P$.
In contrast, all other actors shall be self-powered.

\subsection{Maximal Throughput Periodic Schedules}\label{sec:lp}

To determine the maximal achievable throughput for a hybrid \ac{DFG} implementation, we propose an \ac{LP} formulation as given in~\cref{def:LP}. Solving the \ac{LP} returns a schedule for a DFG $\graph$ (cf.~\cref{lp:graph}) for which it is assumed to be known for each actor $\actor_i \in \SetActors$ whether it is implemented in always-active or self-powered mode.
This is modeled by a binary decision vector $\DecisionVector=(\DecisionVariable{1}, \DecisionVariable{2}, \ldots\DecisionVariable{\SetCard{\SetActors}}) \in \{0,1\}^\SetCard{\SetActors}$ (cf.~\cref{lp:x}).
For each actor $\actor_i \in \SetActors$, the value $\DecisionVariable{i}$ of the decision vector $\DecisionVector$ indicates whether the actor is configured to be always-active ($\DecisionVariable{i} = 0$) or self-powered ($\DecisionVariable{i} = 1$).
As a result of the \ac{LP}, we obtain a periodic schedule $\SchedStartTimes$ (cf.~\cref{lp:schedule-out}) and its period $P \in \RationalsNonNegative$ (cf.~\cref{lp:period-out,lp:period-var}) that shall be minimized (cf.~\cref{lp:objective}).
The schedule $\SchedStartTimes$ contains for each actor $\actor_i \in \SetActors$ its start time $\SchedStartTime{i}$ (cf.~\cref{lp:start-var}), which denotes the time the actor $\actor_i$ fires for the first iteration, and subsequent firings of this actor repeat with a distance of period $P$.
\par
In order to ensure that an actor can begin the next iteration of execution only once its previous execution is completed, a constraint (cf.~\cref{lp:self-loop}) is added for each actor $\actor_i \in \SetActors$ that the period $P$ must be greater than or equal to the execution time ($\ExecutionTime{i}$) of the actor, including the potential overheads of wake-up ($\WakeupDelay{i}$) and shutdown ($\ShutdownDelay{i}$) delays if the actor is self-powered, i.e., $\DecisionVariable{i} = 1$.
Finally, all \ac{DFG} data dependencies have to be respected according to~\cref{eq:data-dependencies}.
Note that if an actor $\actor_i$ is self-powered ($\DecisionVariable{i} = 1$), then not only the actor execution time $\ExecutionTime{i}$ but also the wake-up delay $\WakeupDelay{i}$ must be accounted for, which results in the constraints given in~\cref{lp:channel}.

\pagebreak

\par\noindent\rule{\columnwidth}{0.4pt}\vspace{-3mm}
\begin{listing}{$\texttt{LP}(\graph, \DecisionVector)$ to minimize period $P$\label{def:LP}}
\\\vspace{-2.5mm}\hrule\vspace{1mm}\noindent
  \begin{tabular}{l@{}r@{\;}ll}
    \multicolumn{3}{l}{\textbf{Inputs:}}\\
      \rowlabel{lp:graph} & \quad \ac{DFG} $\graph$             &$=(\SetActors, \SetChannels)$ \\
      \rowlabel{lp:x} & \quad Decision vector $\DecisionVector$ &$=(\DecisionVariable{1}, \DecisionVariable{2}\ldots \DecisionVariable{\SetCard{\SetActors}})$ \\
    \multicolumn{3}{l}{\textbf{Output:}}\\
      \rowlabel{lp:schedule-out} & \quad Schedule $\SchedStartTimes$ & $= (\SchedStartTime{1}, \SchedStartTime{2}, \ldots \SchedStartTime{\SetCard{\SetActors}})$  \\
      \rowlabel{lp:period-out} & \quad Period $P$\\
    \multicolumn{3}{l}{\textbf{LP Objective:}} \\
      \rowlabel{lp:objective}  & \quad $\min$ $P$ \\      
    \multicolumn{3}{l}{\textbf{LP Variables:}}\\
      \rowlabel{lp:period-var} & \quad Period     $P$                 &$\in\RationalsNonNegative$ \\
      \rowlabel{lp:start-var} & \quad Start time $\SchedStartTime{i}$ &$\in \RationalsNonNegative$ & $\forall \actor_i \in \SetActors$ \\
    \multicolumn{3}{l}{\textbf{LP Constraints:}}\\
      \rowlabel{lp:self-loop} & \quad $P$ 
                  &$\ge \ExecutionTime{i}$ \\
      &           &$+\DecisionVariable{i}\cdot\WakeupDelay{i}$\\
      &           &$+\DecisionVariable{i}\cdot\ShutdownDelay{i}$ & $\forall \actor_i \in \SetActors$ \\
      \rowlabel{lp:channel} & \quad $\SchedStartTime{j} + \InitialTokens(\actor_i, \actor_j) \cdot P$ 
                  &$\ge \SchedStartTime{i} + \ExecutionTime{i}$ \\
      &           &$+ \DecisionVariable{i}\cdot\WakeupDelay{i}$ & $\forall (\actor_i, \actor_j) \in \SetChannels$
  \end{tabular}
\vspace{1mm}
\par\noindent\rule{\columnwidth}{0.4pt}\vspace{2mm}
\end{listing}
\par

Solving the \ac{LP} with $\DecisionVector = \vec{0}$, i.e., $\DecisionVariable{i}=0 \; \forall \; \actor_i \in \SetActors$, will determine the minimal period $P_\mathrm{min} = \texttt{LP}(\graph, \vec{0})$ obtained when all actors are in always-active mode.
This period may serve as a lower bound for subsequent \ac{DSE} strategies.
Conversely, solving the \ac{LP} with $\DecisionVector = \vec{1}$, i.e., $\DecisionVariable{i}=1 \; \forall \; \actor_i \in \SetActors$, will determine a periodic schedule of a fully self-powered \ac{DFG} implementation with period $P_\mathrm{max} = \texttt{LP}(\graph, \vec{1})$ that may be used as an upper bound for subsequent \ac{DSE} strategies.
Using~\cref{eq:E_Always_On,eq:E_Self_Powered}, the overall energy consumption of a hybrid implementation  with period $P$ can be determined as given by the following~\cref{eq:totalenergy}.
\begin{equation}
    E(\SetActors,P,\DecisionVector) = \sum_{\actor_i \in \SetActors}{(1-\DecisionVariable{i})\cdot \EnergyAlwaysActiveActor{i} + \DecisionVariable{i}\cdot\EnergySelfPoweredActor{i}}. \label{eq:totalenergy}
\end{equation}

\vspace{-6mm}
\subsection{Minimal Energy Periodic Schedules}\label{sec:milp}

Next, we reformulate the \ac{LP} in~\cref{def:LP} into an \ac{MILP} formulation (see~\cref{def:MILP}) to determine for each actor of a given \ac{DFG} $\graph$ (cf.~\cref{milp:graph}) whether it should run in always-active or in self-powered mode, thereby outputting for a given period $P$ (cf.~\cref{milp:period}) an energy-minimal periodic schedule $\SchedStartTimes$ (cf.~\cref{milp:schedule-out}).
Obviously, while guarding the period constraint $P$, the mode of each actor $\actor_i$ needs to be decided, i.e., to run it in always-active ($\DecisionVariable{i} = 0$) or in self-powered ($\DecisionVariable{i} = 1$) mode (cf~\cref{milp:decision-var}), as encoded by the decision vector $\DecisionVector$, being an output of the \ac{MILP} (cf.~\cref{milp:x-out}).
The objective function (cf.~\cref{milp:objective}) is to minimize the total energy consumption per iteration across all actors, as defined by~\cref{eq:totalenergy}, where for each actor $\actor_i \in \SetActors$, the total energy consumption per iteration is calculated and summed up.
In case a given actor $\actor_i$ is classified as critical ($\DecisionVariable{i} = 0$) in the solution of the \ac{MILP}, the actor must be operated in always-active mode.
Therefore, only execution power $\ExecutionPower{i}$ and idle power $\IdlePower{i}$ are considered in the cost function as energy contribution ($\EnergyAlwaysActiveActor{i}$ acc. to~\cref{eq:E_Always_On}).
Conversely, if the actor is classified as non-critical, the power terms $\EnergySelfPoweredActor{i}$ acc. to~\cref{eq:E_Self_Powered} are added up in the objective function, including shutdown power $\ShutdownPower{i}$, wake-up power $\WakeupPower{i}$, and dormant (sleep) power $\SleepPower{i}$.
Finally, \cref{milp:start-var,milp:self-loop,milp:channel} remain as in the \ac{LP} in~\cref{def:LP}.
\par
\par\noindent\rule{\columnwidth}{0.4pt}\vspace{-3mm}
\begin{listing}{$\texttt{MILP}(\graph, P)$ to minimize energy $E(\SetActors,P,\DecisionVector)$\label{def:MILP}}
\\\vspace{-2.5mm}\hrule\vspace{1mm}\noindent
  \begin{tabular}{l@{}r@{\;}ll}
    \multicolumn{3}{l}{\textbf{Inputs:}}\\
      \rowlabel{milp:graph}  & \quad \ac{DFG} $\graph$                 &$=(\SetActors, \SetChannels)$ \\
      \rowlabel{milp:period} & \quad Period $P$ & $\in \RationalsNonNegative$ \\
    \multicolumn{3}{l}{\textbf{Output:}}\\
      \rowlabel{milp:schedule-out} & \quad Schedule $\SchedStartTimes$ & $= (\SchedStartTime{1}, \SchedStartTime{2}, \ldots \SchedStartTime{\SetCard{\SetActors}})$  \\
      \rowlabel{milp:x-out} & \quad  Decision vector $\DecisionVector$ &$=(\DecisionVariable{1}, \DecisionVariable{2}\ldots \DecisionVariable{\SetCard{\SetActors}})$ \\
    \multicolumn{3}{l}{\textbf{MILP Objective:}} \\
      \rowlabel{milp:objective}  & \quad $\min E$&$\!\!(\SetActors,P,\DecisionVector)$&cf.~\cref{eq:totalenergy}\\
    \multicolumn{3}{l}{\textbf{MILP Variables:}}\\
      \rowlabel{milp:decision-var} & \quad Decision variable $\DecisionVariable{i}$ &$\in \{0, 1\}$ & $\forall \actor_i \in \SetActors$ \\
      \rowlabel{milp:start-var} & \quad Start time $\SchedStartTime{i}$ &$\in \RationalsNonNegative$ & $\forall \actor_i \in \SetActors$ \\
    \multicolumn{3}{l}{\textbf{MILP Constraints:}}\\
      \rowlabel{milp:self-loop} & \quad $P$ 
                  &$\ge \ExecutionTime{i}$ \\
      &           &$+\DecisionVariable{i}\cdot\WakeupDelay{i}$\\
      &           &$+\DecisionVariable{i}\cdot\ShutdownDelay{i}$ & $\forall \actor_i \in \SetActors$ \\
      \rowlabel{milp:channel} & \quad $\SchedStartTime{j} + \InitialTokens(\actor_i, \actor_j) \cdot P$ 
                  &$\ge \SchedStartTime{i} + \ExecutionTime{i}$ \\
      &           &$+ \DecisionVariable{i}\cdot\WakeupDelay{i}$ & $\forall (\actor_i, \actor_j) \in \SetChannels$
  \end{tabular}
\vspace{-3.5mm}
\par\noindent\rule{\columnwidth}{0.4pt}\vspace{2mm}
\end{listing}

\begin{example}
Continuing with our running example, it turns out that for the input $P = P_\textrm{min}$, some actors cannot be put into self-powered mode in any schedule that satisfies this period constraint as they belong to cycles determining the maximal cycle mean (critical cycles), as an increased execution time would lead to a violation of the given period constraint.
Therefore, solving the \ac{MILP} for an input of period $P=P_\textrm{min}=23$ finds a solution in which the actors $\actor_2$, $\actor_3$, $\actor_4$, and $\actor_5$ are identified as critical, and must therefore never be powered down, i.e., returns the decision vector $\DecisionVector = (\DecisionVariable{1}, \DecisionVariable{2}, \DecisionVariable{3}, \DecisionVariable{4}, \DecisionVariable{5}, \DecisionVariable{6}) = (1,0,0,0,0,1)$.
\Cref{fig:hybrid-schedule} shows the corresponding schedule of such a hybrid implementation.
\par
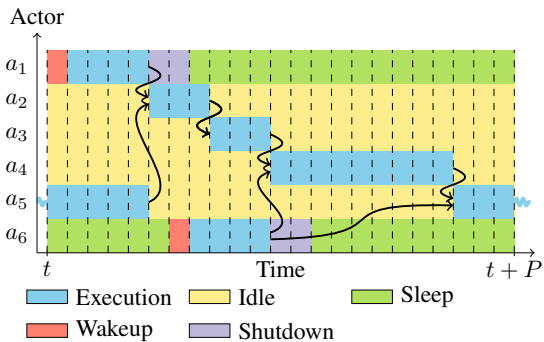
\begin{figure}[ht]
  \vspace{-3mm}
  \centering
  \scalebox{0.9}{\begin{tikzpicture}[xscale=0.3, yscale=0.5]


    

    

    \node[left] at (0.5,0.5) {$\actor_6$};
    \node[left] at (0.5,1.5) {$\actor_5$};
    \node[left] at (0.5,2.5) {$\actor_4$};
    \node[left] at (0.5,3.5) {$\actor_3$};
    \node[left] at (0.5,4.5) {$\actor_2$};
    \node[left] at (0.5,5.5) {$\actor_1$};

    \node[below] at (1,0) {$t$};
    \draw[-] (1,-0.1) -- (1,0.1);
    \node[below] at (24,0) {$t + P$};
    \draw[-] (24,-0.1) -- (24,0.1);
    

    \fill[pastel-yellow] (0 +1,0) rectangle (23 +1,6); 
    \fill[pastel-green] (0 +1,5) rectangle (23 +1,6); 
    \fill[pastel-green] (0 +1,0) rectangle (23 +1,1);

    \fill[pastel-blue]  (1 +1,5) rectangle (5 +1,6);
    \fill[pastel-blue]  (5 +1,4) rectangle (8 +1,5);
    \fill[pastel-blue]  (8 +1,3) rectangle (11 +1,4);
    \fill[pastel-blue]  (11 +1,2) rectangle (20 +1,3);
    \fill[pastel-blue]  (20 +1,1) rectangle (23 +1,2);
    \fill[pastel-blue]  (0 +1,1) rectangle (5 +1,2);
    \fill[pastel-blue]  (7 +1,0) rectangle (11 +1,1);

    \fill[pastel-red]  (0 +1,5) rectangle (1 +1,6); 
    \fill[pastel-lavender]  (5 +1,5) rectangle (7 +1,6); 
    \fill[pastel-red]  (6 +1,0) rectangle (7 +1,1);
    \fill[pastel-lavender]  (11 +1,0) rectangle (13 +1,1);


    \draw[dashed] (1,0) -- (1,6);
    \draw[dashed] (2,0) -- (2,6);
    \draw[dashed] (3,0) -- (3,6);
    \draw[dashed] (4,0) -- (4,6);
    \draw[dashed] (5,0) -- (5,6);
    \draw[dashed] (6,0) -- (6,6);
    \draw[dashed] (7,0) -- (7,6);
    \draw[dashed] (8,0) -- (8,6);
    \draw[dashed] (9,0) -- (9,6);
    \draw[dashed] (10,0) -- (10,6);
    \draw[dashed] (11,0) -- (11,6);
    \draw[dashed] (12,0) -- (12,6);
    \draw[dashed] (13,0) -- (13,6);
    \draw[dashed] (14,0) -- (14,6);
    \draw[dashed] (15,0) -- (15,6);
    \draw[dashed] (16,0) -- (16,6);
    \draw[dashed] (17,0) -- (17,6);
    \draw[dashed] (18,0) -- (18,6);
    \draw[dashed] (19,0) -- (19,6);
    \draw[dashed] (20,0) -- (20,6);
    \draw[dashed] (21,0) -- (21,6);
    \draw[dashed] (22,0) -- (22,6);
    \draw[dashed] (23,0) -- (23,6);
    \draw[dashed] (24,0) -- (24,6);


    \draw[->, thick] (6,5.5) .. controls (8,5) and (4,5) .. (6,4.6);
    \draw[->, thick] (6,1.5) .. controls (8,2) and (4,4) .. (6,4.4);
    \draw[->, thick] (9,4.5) .. controls (11,4) and (7,4) .. (9,3.5);
    \draw[->, thick] (9,4.5) .. controls (11,4) and (7,4) .. (9,3.5);
    \draw[->, thick] (12,3.5) .. controls (14,3) and (10,3) .. (12,2.6);
    \draw[->, thick] (12,0.6) .. controls (14,1) and (10,2) .. (12,2.4);
    \draw[->, thick] (21,2.5) .. controls (23,2) and (19,2) .. (21,1.6);
    \draw[->, thick] (12,0.4) .. controls (19,0.5) and (14,1.5)  .. (21,1.4);

    \draw[snake it, ultra thick, color=pastel-blue] (24,1.5) -- (24.8,1.5);
    \draw[snake it, ultra thick, color=pastel-blue] (0.5,1.5) -- (1,1.5);

    \draw[->] (0.5,0) -- (25,0); 
    \draw[->] (0.5,0) -- (0.5,6.5) node[above] {Actor};
    \node[below] at (12.5,0) {Time};

  \begin{scope}[shift={(0,-1.5)}, every node/.style={anchor=west}]
        \draw[fill=pastel-blue, draw=black] (0,0) rectangle (2,0.4);
        \node at (2,0.2) {Execution};
        \draw[fill=pastel-yellow, draw=black] (8,0) rectangle (10,0.4);
        \node at (10,0.2) {Idle};
        \draw[fill=pastel-green, draw=black] (16,0) rectangle (18,0.4);
        \node at (18,0.2) {Sleep};
        \draw[fill=pastel-red, draw=black] (0,-1) rectangle (2,-0.6);
        \node at (2,-0.8) {Wakeup};
        \draw[fill=pastel-lavender, draw=black] (8,-1) rectangle (10,-0.6);
        \node at (10,-0.8) {Shutdown};
    \end{scope}
    
\end{tikzpicture}}
  \caption{\label{fig:hybrid-schedule}%
    Schedule of a hybrid implementation of the \ac{AEC} network with $P=P_\mathrm{min}=23$.
    Actors $a_1$ and $a_6$ are set to self-powered mode, while the others must be set to always-active mode, to guard this minimal period.}     
\end{figure}
There, actors $\actor_2$, $\actor_3$, $\actor_4$, and $\actor_5$ are always active, while actors $\actor_1$ and $\actor_6$ are disposed to self-powered mode.
The minimal energy determined by the \ac{MILP}, determined by \cref{eq:totalenergy}, amounts to 146\,pJ.
This signifies a savings of $20\,\%$ over the fully always-active case depicted in~\cref{fig:anc-schedule}, without any adverse effect on the period, resp. throughput.
\end{example}

\begin{example}
In contrast, a schedule of the \ac{AEC} network when all actors are self-powered is depicted in~\cref{fig:self-powered-schedule}.
As discussed above, a schedule with period $P_\mathrm{max}=27$ is obtained when solving the $\texttt{LP}(\graph, \vec{1})$, indicating a degradation of $17\,\%$ compared to the schedule depicted in~\cref{fig:hybrid-schedule}.
However, it also results in an $44\,\%$ lower energy per iteration.
Naturally, there are more tradeoff points between the two schedules depicted in~\cref{fig:hybrid-schedule,fig:self-powered-schedule}, which will be determined by our proposed novel \ac{DSE} approach, as described in the next section.
\end{example}

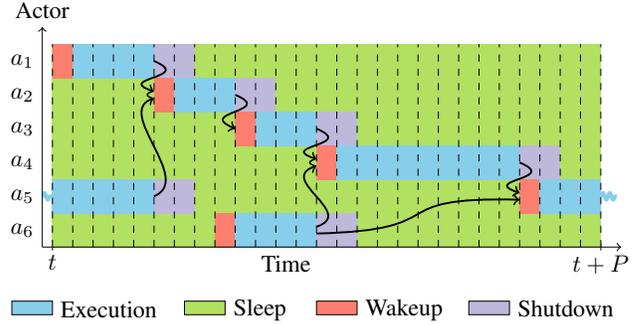
\begin{figure}[ht]
  \centering
  \scalebox{0.9}{\begin{tikzpicture}[xscale=0.3, yscale=0.5]
    \node[left] at (0.5,0.5) {$\actor_6$};
    \node[left] at (0.5,1.5) {$\actor_5$};
    \node[left] at (0.5,2.5) {$\actor_4$};
    \node[left] at (0.5,3.5) {$\actor_3$};
    \node[left] at (0.5,4.5) {$\actor_2$};
    \node[left] at (0.5,5.5) {$\actor_1$};

    \node[below] at (1,0) {$t$};
    \draw[-] (1,-0.1) -- (1,0.1);
    \node[below] at (28,0) {$t + P$};
    \draw[-] (28,-0.1) -- (28,0.1);
    

    \fill[pastel-green]    ( 0+1,0) rectangle (27+1,6); 

    \fill[pastel-red]      ( 0+1,5) rectangle ( 1+1,6);
    \fill[pastel-blue]     ( 1+1,5) rectangle ( 5+1,6);
    \fill[pastel-lavender] ( 5+1,5) rectangle ( 7+1,6);

    \fill[pastel-red]      ( 5+1,4) rectangle ( 6+1,5);
    \fill[pastel-blue]     ( 6+1,4) rectangle ( 9+1,5);
    \fill[pastel-lavender] ( 9+1,4) rectangle (11+1,5);

    \fill[pastel-red]      ( 9+1,3) rectangle (10+1,4);
    \fill[pastel-blue]     (10+1,3) rectangle (13+1,4);
    \fill[pastel-lavender] (13+1,3) rectangle (15+1,4);

    \fill[pastel-red]      (13+1,2) rectangle (14+1,3);
    \fill[pastel-blue]     (14+1,2) rectangle (23+1,3);
    \fill[pastel-lavender] (23+1,2) rectangle (25+1,3);

    \fill[pastel-red]      (23+1,1) rectangle (24+1,2);
    \fill[pastel-blue]     (24+1,1) rectangle (27+1,2);
    \fill[pastel-blue]     ( 0+1,1) rectangle ( 5+1,2);
    \fill[pastel-lavender] ( 5+1,1) rectangle ( 7+1,2);

    \fill[pastel-red]      ( 8+1,0) rectangle ( 9+1,1);
    \fill[pastel-blue]     ( 9+1,0) rectangle (13+1,1);
    \fill[pastel-lavender] (13+1,0) rectangle (15+1,1); 

    \draw[dashed] (1,0) -- (1,6);
    \draw[dashed] (2,0) -- (2,6);
    \draw[dashed] (3,0) -- (3,6);
    \draw[dashed] (4,0) -- (4,6);
    \draw[dashed] (5,0) -- (5,6);
    \draw[dashed] (6,0) -- (6,6);
    \draw[dashed] (7,0) -- (7,6);
    \draw[dashed] (8,0) -- (8,6);
    \draw[dashed] (9,0) -- (9,6);
    \draw[dashed] (10,0) -- (10,6);
    \draw[dashed] (11,0) -- (11,6);
    \draw[dashed] (12,0) -- (12,6);
    \draw[dashed] (13,0) -- (13,6);
    \draw[dashed] (14,0) -- (14,6);
    \draw[dashed] (15,0) -- (15,6);
    \draw[dashed] (16,0) -- (16,6);
    \draw[dashed] (17,0) -- (17,6);
    \draw[dashed] (18,0) -- (18,6);
    \draw[dashed] (19,0) -- (19,6);
    \draw[dashed] (20,0) -- (20,6);
    \draw[dashed] (21,0) -- (21,6);
    \draw[dashed] (22,0) -- (22,6);
    \draw[dashed] (23,0) -- (23,6);
    \draw[dashed] (24,0) -- (24,6);
    \draw[dashed] (25,0) -- (25,6);
    \draw[dashed] (26,0) -- (26,6);
    \draw[dashed] (27,0) -- (27,6);
    \draw[dashed] (28,0) -- (28,6);

    \draw[->, thick] (6,5.5) .. controls (8,5) and (4,5) .. (6,4.6);
    \draw[->, thick] (6,1.5) .. controls (8,2) and (4,4) .. (6,4.4);
    \draw[->, thick] (10,4.5) .. controls (12,4) and (8,4) .. (10,3.5);
    \draw[->, thick] (14,3.5) .. controls (16,3) and (12,3) .. (14,2.6);
    \draw[->, thick] (14,0.6) .. controls (16,1) and (12,2) .. (14,2.4);
    \draw[->, thick] (24,2.5) .. controls (26,2) and (22,2) .. (24,1.6);
    \draw[->, thick] (14,0.4) .. controls (21,0.5) and (17,1.5)  .. (24,1.4);

    \draw[snake it, ultra thick, color=pastel-blue] (28,1.5) -- (28.8,1.5);
    \draw[snake it, ultra thick, color=pastel-blue] (0.5,1.5) -- (1,1.5);

    \draw[->] (0.5,0) -- (29,0); 
    \draw[->] (0.5,0) -- (0.5,6.5) node[above] {Actor};
    \node[below] at (12.5,0) {Time};

    \begin{scope}[shift={(-4,-2)}, every node/.style={anchor=west}]
        \draw[fill=pastel-blue, draw=black] (3,0) rectangle (5,0.4);
        \node at (5,0.17) {Execution};
        
        \draw[fill=pastel-green, draw=black] (11.5,0) rectangle (13.5,0.4);
        \node at (13.5,0.13) {Sleep};
        
        \draw[fill=pastel-red, draw=black] (18,0) rectangle (20,0.4);
        \node at (20,0.13) {Wakeup};
        
        \draw[fill=pastel-lavender, draw=black] (25.5,0) rectangle (27.5,0.4);
        \node at (27.5,0.17) {Shutdown};
    \end{scope}
    
\end{tikzpicture}}
  \caption{\label{fig:self-powered-schedule}%
    Schedule of the \ac{AEC} network with $P=P_\mathrm{max}=27$, obtained when all actors are self-powered.}
\end{figure}

\section{Design Space Exploration}\label{sec:dse}

From the previously introduced examples, we can conclude that there can be quite a huge range of tradeoff solutions between the choice of period $P$ (i.e., throughput) and the achievable energy $E$ of resulting \ac{DFG} implementations.
For performing a \acf{DSE} of these tradeoffs, we next discuss two alternative brute-force search strategies before proposing \emph{Hop and Skip}, a novel and efficient \ac{DSE} strategy.
\par
For the explanation of the individual strategies, we use as an example a \acf{SDF}~\cite{E.Lee87SDF} graph called Samplerate from~\cite{sdf3}.%
\footnote{An \ac{SDF} graph can be unrolled by the algorithm from~\cite{E.Lee86ACH} to generate an equivalent marked graph that adheres to the semantics described in this paper.}

\subsection{Decision Variable Sweep\label{sec:dse-xs}}

A (true) Pareto front of energy/troughput solutions can be determined by an exhaustive sweep across all $2^\SetCard{\SetActors}$ different combinations of the decision vector $\DecisionVector \in \{0,1\}^\SetCard{\SetActors}$ (see~\crefrange{alg:dse-xs:loop-start}{alg:dse-xs:loop-end} of~\cref{alg:dse-xs}), and apply the \ac{LP} from~\cref{def:LP} to determine the minimum sustainable period $P$ (see~\cref{alg:dse-xs:loop-body}) and the corresponding energy per period (see~\cref{alg:dse-xs:loop-add-ep}) for each configuration.
A decision-variable sweep requires solving an \ac{LP} for each decision vector $\DecisionVector$, with an \ac{LP} being solvable in polynomial time.
However, the number of configurations grows exponentially in the number of actors $\SetCard{\SetActors}$.
Hence, for large networks, such a strategy is not feasible.
\par
\begin{algorithm}
  \caption{Decision Variable Sweep $\texttt{DSE\_XS}(\graph)$}\label{alg:dse-xs}
  \DontPrintSemicolon
  \SetKwFunction{FLP}{LP}
  \SetKwFunction{FILP}{MILP}
  \KwIn  {\hspace{2.5mm}\parbox{25mm}{\ac{DFG}} $\graph =(\SetActors,\SetChannels)$}
  \KwOut {\parbox{24mm}{Explored points} $EP$}
  \tcp{Initialize explored points set} 
  $EP \gets \emptyset$ \; \label{alg-xs:dse:clear-ep}
  \tcp{Sweep over the design space $\{0,1\}^\SetCard{\SetActors}$}
  \ForEach {$\DecisionVector \in \{0,1\}^\SetCard{\SetActors}$\label{alg:dse-xs:loop-start}}{
    $P \gets$ \FLP{$\graph, \DecisionVector$}\; \label{alg:dse-xs:loop-body}
    $EP \gets EP \cup \{ (P,E(\SetActors,P,\DecisionVector),\DecisionVector) \}$\;  \label{alg:dse-xs:loop-add-ep}
  }\label{alg:dse-xs:loop-end}
  \Return $EP$\;\label{alg:dse-xs:return-ep}
\end{algorithm}

\begin{example}
  \Cref{fig:dse-xs} shows the objective space covered by the decision-variable sweep $\texttt{DSE\_XS}(\graph)$.
  On the x-axis, the period $P$ is shown, while the y-axis represents the total energy $E$ per period of a found schedule.
  Each of the $2^6=64$ blue circles represents a unique design configuration $\DecisionVector \in \{0,1\}^\SetCard{\SetActors}$, while the red 'o's mark non-dominated points.
  For the Samplerate benchmark, the exploration using \cref{alg:dse-xs} took a total of 2.23 CPU seconds on an Intel(R) Core(TM) i7-14700K machine.
\end{example}

\begin{theorem}\label{theorem:pareto}
  Given a \ac{DFG} $\graph$, the set of non-dominated points of the set $EP$ explored by the decision-variable sweep strategy $\texttt{DSE\_XS(\graph)}$ coincides with the (true) Pareto front of energy/throughput solutions.
\end{theorem}

\begin{proof}
  First, observe that the energy cost function of each actor in~\cref{eq:totalenergy} is linear in $P$, with a positive gradient for both cases of $\actor_i$ being either executed in always-active mode or in self-powered mode (corresponding to $\IdlePower{i}$ in \cref{eq:E_Always_On} and $\SleepPower{i}$ in \cref{eq:E_Self_Powered}).
  Therefore, the overall cost function $E(\SetActors,P,\DecisionVector)$ is strictly increasing with increasing values of $P$; hence, minimizing $P$ for a given configuration $\DecisionVector$ also minimizes $E(\SetActors,P,\DecisionVector)$ for the same $\DecisionVector$.
  As we may assume that the \ac{LP} given in~\cref{def:LP} is able to determine a feasible schedule of minimal period $P$ supporting $\DecisionVector$, the energy $E(\SetActors,\texttt{LP}(\graph, \DecisionVector),\DecisionVector)$ is the minimal energy for configuration $\DecisionVector$.
  \par
  Thus, the point $(\texttt{LP}(\graph, \DecisionVector),E(\SetActors,\texttt{LP}(\graph, \DecisionVector),\DecisionVector),\DecisionVector)$ is a non-dominated point among all solutions with the same configuration $\DecisionVector$.
  As $EP = \{(\texttt{LP}(\graph, \DecisionVector),E(\SetActors,\texttt{LP}(\graph, \DecisionVector),\DecisionVector),\DecisionVector) \mid \DecisionVector\in\{0,1\}^\SetCard{\SetActors}\}$ contains all these candidates for all configurations, we can conclude it also contains the Pareto-optimal points as the set of non-dominated points within $EP$.
\end{proof}

\subsection{Period Sweep\label{sec:period-sweep}}

As the number of configurations grows exponentially with the number of actors, rendering the decision-variable sweep strategy infeasible for complex networks, there is a need to investigate alternative \ac{DSE} strategies.
A simple exhaustive search strategy is to sweep the space of all integral periods and determine an energy-minimal implementation for each period simply by solving the \ac{MILP} as introduced in~\cref{def:MILP}.
The period sweep starts from the minimum period $P_{\mathrm{min}}$ yielded by the \ac{LP} with all actors configured to always-active, and ends with the maximum period $P_{\mathrm{max}}$, which the \ac{LP} similarly determines with all actors configured to self-powering mode.
Thus, a sweep solves the \ac{MILP} in~\cref{def:MILP} exactly $P_{\mathrm{max}} - P_{\mathrm{min}} + 1$ times, as elaborated in \cref{alg:dse-ps}.

\begin{algorithm}
  \caption{Period Sweep $\texttt{DSE\_PS}(\graph)$}\label{alg:dse-ps}
  \DontPrintSemicolon
  \SetKwFunction{FLP}{LP}
  \SetKwFunction{FILP}{MILP}
  \KwIn  {\hspace{2.5mm}\parbox{25mm}{\ac{DFG}} $\graph =(\SetActors,\SetChannels)$}
  \KwOut {\parbox{24mm}{Explored points} $EP$}
  \tcp{Always-active schedule period}
  $P_{\textrm{min}} \gets $\FLP{$\graph, \vec{\mathbf{0}}$} \; \label{alg:dse-ps:pmin}
  \tcp{Self-powered schedule period}
  $P_{\textrm{max}} \gets $\FLP{$\graph, \vec{\mathbf{1}}$} \;  \label{alg-ps:dse:pmax}
  \tcp{Initialize explored points set}
  $EP \gets \emptyset$ \;  \label{alg:dse-ps:clear-ep}
  \tcp{Sweep periods between $P_{\textrm{min}}$ and $P_{\textrm{max}}$}
  \ForEach{$P \in [P_{\textrm{min}},P_{\textrm{max}}] \cap \NonNegativeIntegers$\label{alg:dse-ps:loop-start}}{
    $\DecisionVector \gets$ \FILP{$\graph, P$}\; \label{alg:dse:loop-milp}
    $EP \gets EP \cup \{ (P,E(\SetActors,P,\DecisionVector),\DecisionVector) \}$\;  \label{alg:dse-ps:loop-add-ep}
  }\label{alg:dse-ps:loop-end}
  \Return $EP$\;\label{alg:dse-ps:return-ep}
\end{algorithm}

\begin{example}
  \Cref{fig:dse-ps} depicts the objective space covered by the above-mentioned period-sweep strategy applied to the network Samplerate~\cite{sdf3} with a total of 6 \ac{SDF} actors.
  The green diamond represents the configuration where all actors are configured to always-active mode, with a resulting period of $P = P_{\textrm{min}} = 160$, while the green triangle represents configurations where all the actors are self-powered, with a resulting period $P =_{\textrm{max}}=320$.
  The blue circles represent all points in the objective space determined each by one \ac{MILP} solver run, thus a total of $161$ \ac{MILP} solver runs.
  Finally, the non-dominated points are identified and marked by red 'o's.
  However, the whole experiment took a total exploration time of 23.2 CPU seconds. 
  Therefore, if the design space of the decision vector $\DecisionVector \in \{0,1\}^\SetCard{\SetActors}$ is rather small (small node number \acp{DFG}), a decision variable sweep might be faster.
\end{example}

Note that often, integer periods are desired or even required for complexity reasons.
For example, when realized in hardware, integer periods can be associated with clock cycles.
If rational periods are allowed, the above sweep over only integer periods might not be able to find any potential Pareto-optimal points with rational periods $P \in \RationalsNonNegative$.

\subsection{Hop \& Skip \ac{DSE} Strategy\label{sec:dse-sub}}

The aforementioned period sweep strategy may also require significant computational effort, with $P_{\textrm{max}}-P_{\textrm{min}}+1$ \ac{MILP} solver runs, underscoring a need for a faster \ac{DSE} strategy.
In the following, we propose an efficient exploration algorithm named Hop and Skip (\gls{hnf}).
Similar to \cref{alg:dse-ps}, it starts by determining the period search limits $P_{\mathrm{min}}$ and $P_{\mathrm{max}}$, see~\cref{alg:dse:pmin,alg:dse:pmax} of the following~\cref{alg:dse}.
These bounds are determined using the \ac{LP} with all actors configured to always-active (see~\cref{alg:dse:pmin}) and self-powering mode (see~\cref{alg:dse:pmax}), identical to the period bounds as described in~\cref{sec:period-sweep}.
\par
The set of explored points $EP$ is initially cleared (see~\cref{alg:dse:clear-ep}).
Subsequently, the exploration of the design space is started with the period $P =  P_{\mathrm{max}}$ in~\cref{alg:dse:start-from-pmax}.
The search space of potentially non-dominated candidates is then efficiently traversed from higher values of $P$ to lower values, see the while loop from~\crefrange{alg:dse:loop-start}{alg:dse:loop-end}, terminating when the period $P$ is smaller than $P_{\mathrm{min}}$.
\par
In the loop body, the \ac{MILP} is called (see~\cref{alg:dse:loop-milp}) to first determine a decision vector $\DecisionVector$, corresponding to \emph{hopping} to a minimum-energy design configuration that sustains the period $P$.
For the determined configuration with decision vector $\DecisionVector$, the minimum achievable period $P'$ is then determined by solving the \ac{LP} (see~\cref{alg:dse:loop-lp}), potentially reducing the period $P$ to $P'$.
This step, similar to \emph{skipping} over all candidate values from $P$ to $P'$, allows for a more efficient traversal over the period space, without requiring to check all integer periods like in \cref{alg:dse-ps}, or all $2^\SetCard{\SetActors}$ configurations like in \cref{alg:dse-xs}.
In~\cref{alg:dse:loop-add-ep}, this newly explored point $(P',E(\SetActors,P,\DecisionVector),\DecisionVector)$ is added to the set of explored points $EP$.
Then, the period $P'$ is decremented by a user-defined threshold value $\varepsilon$ to determine the next lower $P$ value in~\cref{alg:dse:loop-decrement-period}, and the loop is iterated until the end is reached when $P < P_{\mathrm{min}}$.
Finally, the set of explored points $EP$ is returned (see~\cref{alg:dse:return-ep}).
  
\begin{algorithm}
  \caption{\gls{hnf} Strategy $\texttt{H\&S}(\graph)$}\label{alg:dse}
  \DontPrintSemicolon
  \SetKwFunction{FLP}{LP}
  \SetKwFunction{FILP}{MILP}
  \KwIn  {\hspace{2.5mm}\parbox{25mm}{\ac{DFG}} $\graph =(\SetActors,\SetChannels)$}
  \KwOut {\parbox{24mm}{Explored points} $EP$}
  \tcp{Always-active schedule period}
  $P_{\textrm{min}} \gets $\FLP{$\graph, \vec{\mathbf{0}}$} \; \label{alg:dse:pmin}
  \tcp{Self-powered schedule period}
  $P_{\textrm{max}} \gets $\FLP{$\graph, \vec{\mathbf{1}}$} \;  \label{alg:dse:pmax}
  \tcp{Initialize explored points set}
  $EP \gets \emptyset$ \;  \label{alg:dse:clear-ep}
  $P \gets P_{\textrm{max}}$\; \label{alg:dse:start-from-pmax}
  \While{$P \geq P_{\textrm{min}}$\label{alg:dse:loop-start}}{
    $\DecisionVector \gets$ \FILP{$\graph, P$} \tcp*{Hop step}  \label{alg:dse:loop-milp}
    $P' \gets$ \FLP{$\graph,\DecisionVector$} \;   \label{alg:dse:loop-lp} 
    $EP \gets EP \cup \{ (P',E(\SetActors,P',\DecisionVector),\DecisionVector) \}$\;  \label{alg:dse:loop-add-ep}
    $P\gets P'-\varepsilon$ \tcp*{Skip step} \label{alg:dse:loop-decrement-period}
  }\label{alg:dse:loop-end}
  \Return $EP$\;\label{alg:dse:return-ep}
\end{algorithm}

\begin{figure}
  \centering
  \vspace{-2mm}
  \subfloat[Decision variable sweep]{
    \hspace{-3mm}
    \input{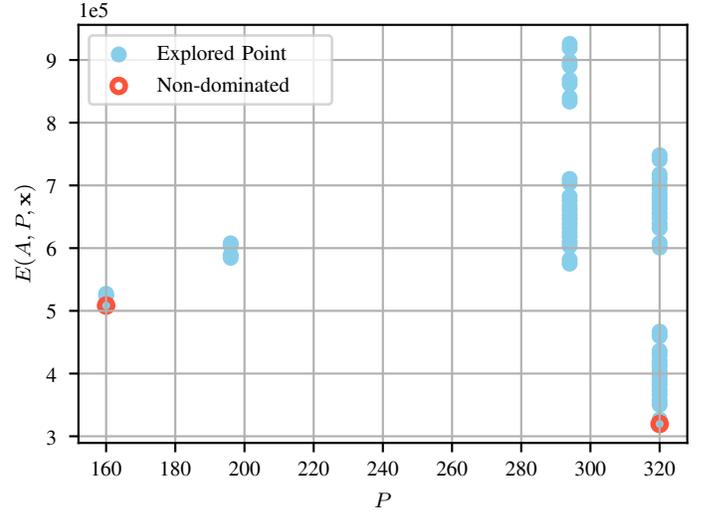}
    \label{fig:dse-xs}
  }

  \subfloat[Period sweep]{
      \hspace{-3mm}
      \input{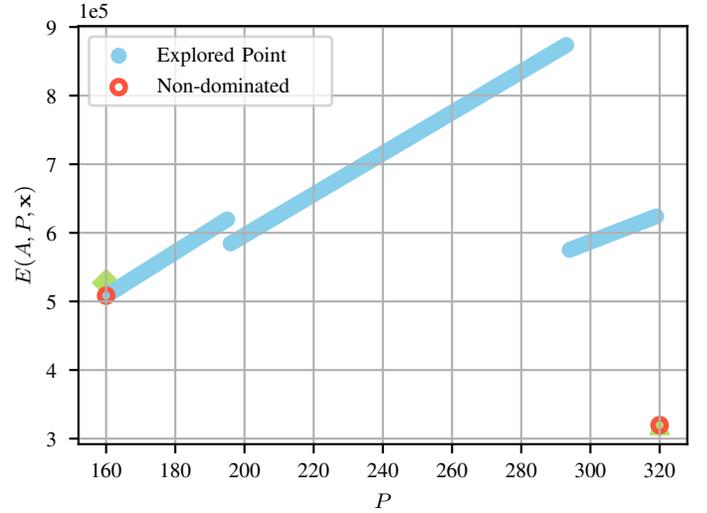}
      \label{fig:dse-ps}
  }

  \subfloat[\gls{hnf}]{
    \hspace{-3mm}
    \input{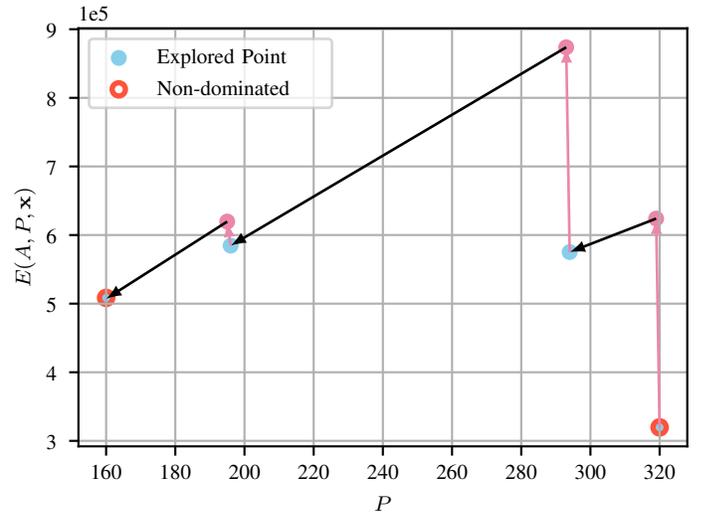}
    \label{fig:dse-new}
  }
  \caption{Explored design points showing tradeoffs between $P$ and $E$ for the Samplerate benchmark~\cite{sdf3} yielded by different \ac{DSE} strategies.
    The green diamond and the triangle correspond to the configurations in which all actors are always active ($\DecisionVector = \vec{0}$), respectively, self-powered ($\DecisionVector = \vec{1}$).}
  \label{fig:dse}
\end{figure}

\begin{example}
  \Cref{fig:dse-new} shows the points explored by the \gls{hnf} strategy.
  The first loop iteration starts with $P=P_{\textrm{max}}=320$. Solving the \ac{MILP} finds the fully self-powered configuration $\DecisionVector=\vec{1}$, and solving the \ac{LP} with $\DecisionVector=\vec{1}$ does not find a lower period $P'<P$.
  In the next loop iteration, with $\varepsilon=0.1$, $P=P_{\textrm{max}}-0.1$, the \ac{MILP} is employed to guard the period and returns the lowest energy configuration $\DecisionVector$.
  This \ac{MILP} step can be visualized by the pink \textit{hop} arrow to the pink circle with objective vector $(P,E) = (P_\textrm{max}-\varepsilon,E(\SetActors,P_\textrm{max}-\varepsilon),\DecisionVector)$.
  The \ac{LP} is then subsequently employed to \textit{skip} to the lowest period $P'$, visualized by the black arrows, supported by the determined configuration $\DecisionVector$.
  The new design point with objectives $(P',E(\SetActors,P',\DecisionVector))$ is appended to $EP$, and the period is then decremented again by $\varepsilon$.
  The process is iterated, with hops and skips yielding design points marked, respectively, by the pink and blue circles, until finally $P < P_{\textrm{min}}$.
  \par
  It is notable to observe that in this example, although significantly fewer points (i.e., only 4) were needed to be evaluated by our proposed \ac{DSE} strategy, both Pareto points were found.
  The \gls{hnf} exploration required 0.83 CPU seconds, which implies a 27.83$\times$ speedup over the period sweep strategy and a 2.67$\times$ speedup over the decision variable sweep strategy.
\end{example}

\section{Experiments and Results}\label{sec:results}

First, in~\cref{sec:AEC-eval}, we introduce a real-world case study to analyze achievable energy savings and achievable throughputs of dataflow networks where (a) all actors are always-active, (b) all actors are self-powered, and (c) all optimal energy and throughput tradeoff points, as found by the proposed \ac{DSE} strategy.
Then, \cref{sec:hnf-eval} provides experimental evaluations on a set of \ac{DFG} graph benchmarks and $100$ random graphs, demonstrating that our proposed \gls{hnf} strategy significantly reduces exploration time compared to brute-force sweeps.

\subsection{Energy and Throughput Tradeoffs for an \ac{AEC} Network}\label{sec:AEC-eval}

In the following, we analyze a real-life data flow application called \ac{AEC} network as presented already in our introductory example in~\cref{fig:echo-cancellation}.
Based on real hardware designs synthesized for each actor of the dataflow flow specification, we  were able to extract also realistic values for power, execution times, and shutdown as well as wake-up delays for each actor of the network using vector-based simulation on the gate-level netlists of each actor implementation using QuestaSim~\cite{questa}.
The application has been initially modeled using SysteMoC~\cite{FalkHTZ17}, realizing the FunState model of computation for actor networks~\cite{ThieleSZET99}.
The gate-level netlists were generated using \ac{HLS}~\cite{catapultC} and \ac{RTL}~\cite{synopsys} synthesis tools, utilizing a 28\,nm process~\cite{synopsys-edk}.
Finally, the power values were obtained using PrimePower~\cite{synopsys}, while the execution times and delays were observed from the vector traces.

For all self-powering implementations, we used clock-gating as the power-saving strategy as described in~\cite{karim2024-selfpoweringDFGs}.
We observed a constant shutdown delay $\ShutdownDelay{i}=2$ and wake-up delay $\WakeupDelay{i}=1$ clock cycles, irrespective of the actor.
Note that these delay values might usually vary between actors and also depend on the chosen power-saving strategy, such as power-gating~\cite{karim2025-selfpoweringDFGs-mbmv,darne2026-selfpoweringDFGs-FeMFET-state-retention}.
The execution times $\ExecutionTime{i}$ were obtained by synthesis and used already in the running example introduced in~\cref{sec:fundamentals}.
\par
For \ac{DSE} of energy and throughput tradeoffs, we then applied the $\DecisionVector$ Sweep strategy described in~\cref{alg:dse-xs} to the \ac{AEC} network, using the parameter values obtained by the hardware synthesis.
Based on our proof in~\cref{sec:dse-xs}, this algorithm is able to find the (true) Pareto front.
In~\cref{fig:pareto-plot}, all $2^6$ explored solutions are shown as blue dots, of which 5 are marked with red 'o's to denote that they belong to the Pareto front.
\par
\begin{figure}[ht]
    \centering
    \resizebox{ \columnwidth }{!}{\input{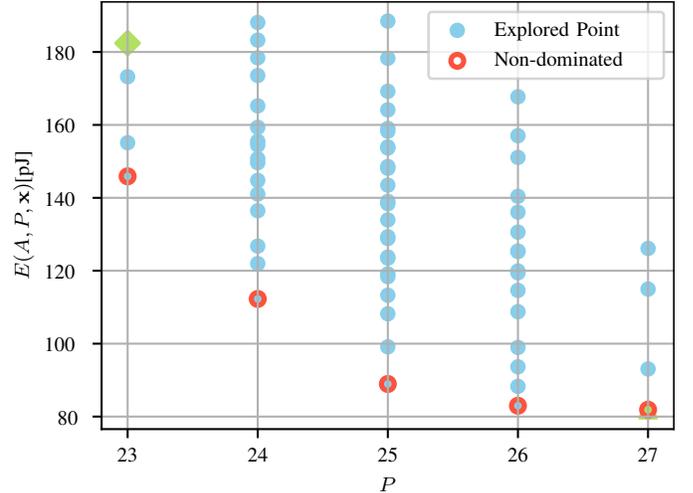}} 
    \caption{Energy $E$ per period $P$ trafeoffs for different hybrid implementations of the \ac{AEC} network.
    The green diamond represents the implementation where all the actors are always active (i.e., $\DecisionVector = \vec{0}$), while the green triangle represents the fully self-powered implementation (i.e., $\DecisionVector = \vec{1}$), which belongs even to the Pareto front of the \ac{AEC} network.}
    \label{fig:pareto-plot}
\end{figure}

Moreover, for comparison, \cref{fig:pareto-plot} also depicts the implementation with all actors being always active as a green diamond, achieving the minimal period $P_\textrm{min} = 23$ and determined as the solution of the $\texttt{LP}(\graph,\vec{0})$ as discussed in~\cref{sec:lp}.
The energy per period is then determined using~\cref{eq:totalenergy} as $E(\SetActors,P_\textrm{min},\vec{0})=182$\,pJ.
However, this solution is not Pareto-optimal as there exists a hybrid configuration $\DecisionVector = (1, 0, 0, 0, 0, 1)$ (see~\cref{fig:hybrid-schedule}), which also sustains the minimal period $P_\textrm{min}$ but only consumes $E(\SetActors,P_\textrm{min},\DecisionVector)=146$\,pJ, thus achieving 20\,\% of energy savings without any throughput degradation.
\par
Finally, if we allow for an increase in the period by $4$ time units, the fully self-powered implementation (i.e., $\DecisionVector=\vec{1}$) becomes feasible, with a period $P_\mathrm{max} = 27$ and determined as the solution of the $\texttt{LP}(\graph,\vec{1})$.
The related energy per period, determined using~\cref{eq:totalenergy}, amounts to $E(\SetActors,P_\textrm{max},\vec{1})=82$\,pJ, accounting for $55\,\%$ savings over the aforementioned fully always-active implementation at an increase of just 17\,\% in the period.

\subsection{Evaluation of Exploration Strategies}\label{sec:hnf-eval}

To elaborate and compare the tradeoffs of the three presented exploration strategies (i.e., \crefrange{alg:dse-xs}{alg:dse}) in detail for larger and more complex benchmarks, we use (a) existing benchmark \acp{DFG} as well as (b) \acp{DFG} generated by a random graph generator, both provided as part of the tool and benchmark suite called "SDF$^3$ - SDF For Free"~\cite{sdf3}.

\begin{enumerate}
  \item[(a)]{SDF$^3$ Benchmarks:}
    We first augmented and analyzed the application networks presented in the "SDF$^3$ - SDF For Free"~\cite{sdf3} benchmark suite.
    These benchmarks represent the following real-world applications, categorized by their topology:
    \begin{itemize}
        \item Cyclic Topology: H.263 Encoder and Modem networks.
        \item Feed-forward Topology:  MP3 Playback, Satellite, Samplerate, MP3 Decoder (configured for granule and block parallelism) networks.
    \end{itemize}
    Since the original SDF$^3$ application networks lack any power and delay annotations, we created needed annotations by performing the following augmentations for each graph:
    First, we assigned to each actor $\actor_i \in \SetActors$ normalized values for power, i.e., $\ExecutionPower{i} = 1$, {$\ShutdownPower{i} = 0.5$, $\WakeupPower{i} = 0.5$, $\SleepPower{i} = 0.1$, and $\IdlePower{i} = 0.9$}, as well as constant delays, i.e., $\ShutdownDelay{i} = 2$ and $\WakeupDelay{i} = 1$.
    In contrast, the execution times $\ExecutionTime{i}$ for each actor were given and extracted directly from the SDF$^3$ application graphs.
    Second, in order to obtain actor-specific power values, we randomly selected a scaling factor $\psi_i \in [1,8]$ for each actor $\actor_i \in \SetActors$ and divided the actor's execution time $\ExecutionPower{i}$ by $\psi_i$, while also multiplying its power $\ExecutionPower{i}$ by $\psi_i$.
    Thus, the actor's execution energy $\ExecutionTime{i}\cdot\ExecutionPower{i}$ remains unaffected.
    Moreover, we also scaled the other power model parameters $\ShutdownPower{i}$, $\WakeupPower{i}$, $\SleepPower{i}$, and $\IdlePower{i}$ by multiplying them with $\psi_i$.

  \item[(b)]{Random Graphs:}
    We also generated a set of 100 \ac{SDF} graphs using the SDF$^3$ random graph generator.
    Each \ac{SDF} graph contains (ideally) 15 actors, and the generation is parameterized so that each actor has an average degree (number of connected channels) of 3, with a variance of 3.
    The consumption and production rates were varied between 1 and 20, with an average of 3 and a variance of 6.
    However, the consumption and production rates were allowed to exceed these limits to ensure that the repetition vector sum of each \ac{SDF} graph is 250 actor firings per period.
    These generated graphs were augmented with randomized delay and power values, creating a statistically diverse set of scenarios.
\end{enumerate}

For the analysis results that will be shown in \cref{tbl:dse-comp,tbl:appendix1}, note that \ac{SDF} graphs  having a total of $\SetCard{\SetActors_{\textrm{SDF}}}$ nodes are converted to marked graphs to be able to apply our analysis techniques.
A conversion described in~\cite{E.Lee86ACH} is applied, which usually leads to a higher number $\SetCard{\SetActors}$.
For the conversion, we also utilize the SDF$^3$ tool suite~\cite{sdf3}.
Fortunately, the decision vector $\DecisionVector$ (and thus the design space) does not increase in size with the unfolded graph.

\textbf{Experimental Setup and Metrics:}
Our forthcoming evaluation consists of:
(a) comparing the non-dominated front generated by the period sweep ($P$ Sweep) and the proposed \gls{hnf} approach, against the true Pareto front returned by the decision-variable sweep ($\DecisionVector$ Sweep); and
(b) comparing the exploration times of each approach.
\par
In general, a \ac{MOP}, in our case, for energy and throughput, does not have a single optimal solution due to multiple conflicting objectives.
Instead, there exists a set of Pareto-optimal solutions.
The set of all such solutions is known as the Pareto front.
The true Pareto front of the \ac{MOP} that is considered in this paper as a reference can be found by applying the $\DecisionVector$ Sweep strategy (\cref{alg:dse-xs}) as was proven in~\cref{theorem:pareto}.
This set acts as the reference Pareto front and is named $\ParetoRef$ in the following.
The quality of each approach and each network can then be evaluated by comparing the sets of non-dominated solutions $\ParetoApp$ of each of the two other \ac{DSE} strategies introduced in \cref{alg:dse-ps,alg:dse} with $\ParetoRef$.
For the comparison of explored non-dominated sets of points, we use an indicator known in literature as \emph{hypervolume}~\cite{Guerrero:2022}.
According to~\cite{hyper}, we first normalize the reference Pareto front $\ParetoRef$ and each non-dominated front $\ParetoApp$ found by $P$ sweep and \gls{hnf} to only contain objective values between zero and one, i.e., $\ParetoRef, \ParetoApp \subset [0,1]^d$ with $d=2$ in our case of the two objectives of energy $E$ and period $P$.
Here, an objective value of zero corresponds to the best (i.e., minimal) value ever found for this objective by all approaches for a given application.
Conversely, an objective value of one corresponds to the worst (i.e., maximal) value ever encountered.
This normalization ensures that each objective is weighted equally in the hypervolume quality measure.
Then, given a (normalized) Pareto front $\ParetoFront \subset [0,1]^d$, the hypervolume of $\ParetoFront$ is the measure of the region weakly dominated\footnote{A point $p \in \Reals^d$ \emph{weakly dominates} a point $q \in \Reals^d$ if $p_i\leq q_i$ for all $1\leq i \leq d$.} by $\ParetoFront$ and bounded above by the reference point $\mathbf{1}$.
\vspace{-5mm}
\par
\begin{equation}\label{eq:hypervolume}
  \textup{hypervolume}(\ParetoFront) = \Lambda( \{ q\in [0,1]^d \mid \exists p \in \ParetoFront : p \leq q\} )
\end{equation}
In~\cref{eq:hypervolume}, $\Lambda(\cdot)$ denotes the Lebesgue measure~\cite{ciesielski1989good}.
The greater hypervolume($S$) is, the better a Pareto front approximation $\ParetoFront$ is considered to be.

\textbf{Discussion of Results:}
\Cref{tbl:dse-comp} provides a comprehensive overview of the performance results of the three introduced \ac{DSE} methods $\DecisionVector$ Sweep, $P$ Sweep, and \gls{hnf}.
The number of actors in the base \ac{SDF} is denoted by $\SetCard{\SetActors_\textrm{SDF}}$, while $\SetCard{\SetActors}$ represents the number of actors of the corresponding unrolled marked graph.
The exploration times in seconds required by each of the strategies are reported,
along with the speedup achieved using the proposed \gls{hnf} strategy compared to both the $\DecisionVector$ Sweep and the $P$ Sweep strategies.
Furthermore, the \emph{hypervolume ratio} $HV$, as defined by~\cref{eq:hypervolume-rel-avg}, is used as a quality indicator of each explored non-dominated set of points $\ParetoApp$.
Note that from the fact that the hypervolume of any non-dominated set cannot exceed the hypervolume of the Pareto front $\ParetoRef$, we always have $HV \leq 1$. 
Moreover, if $HV = 1$, we infer $\ParetoApp = \ParetoRef$.
\begin{equation}\label{eq:hypervolume-rel-avg}
  HV = \frac{\textup{hypervolume}(\ParetoApp)}{\textup{hypervolume}(\ParetoRef)}
\end{equation}
\par
\begin{table*}[h!]
    \centering
    \begin{tabular}{|c|c|c|c|c|c|c|c|c|c|}
    \hline
    \multirow{2}{*}{Network}   & \multirow{2}{*}{$|\SetActors_{\textrm{SDF}}|$} & \multirow{2}{*}{$|\SetActors|$} & \multicolumn{3}{c|}{Exploration Time [s]}                 & \multicolumn{2}{c|}{Hypervolume Ratio} & \multicolumn{2}{c|}{Speedup} \\
                                                                                  \cline{4-10}
                               &    &                   & $\DecisionVector$ Sweep           & $P$ Sweep                        & \gls{hnf}& PS/XS     & \gls{hnf}/XS           & vs $\DecisionVector$ Sweep  & vs $P$ Sweep        \\\hline
        {Satellite}            & 22 & 4515              & Timeout                           & 1495.84                          & 4.28     & -         & -                      &  -                          & 348.93              \\
        Modem                  & 16 & 48                & 264.65                            & 1.51                             & 0.18     & 1         & 1                      & 1452.39$\times$             & 8.30$\times$        \\
        MP3 Decoder (Block)    & 14 & 911               & 628.10                            & 1.54                             & 0.42     & 1         & 1                      & 1489.11$\times$             & 3.65$\times$        \\
        MP3 Decoder (Granule)  & 14 & 27                & 42.50                             & 1.33                             & 0.15     & 1         & 1                      & 289.63$\times$              & 9.09$\times$        \\
        Samplerate             & 6  & 612               & 2.23                              & 23.21                            & 0.83     & 1         & 1                      & 2.67$\times$                & 27.83$\times$       \\
        H263 Encoder           & 5  & 201               & 0.40                              & 1.43                             & 0.30     & 1         & 1                      & 1.36$\times$                & 4.85$\times$        \\
        MP3 Playback           & 4  & 10601             & 20.18                             & 100533.41                        & 42.63    & 1         & 1                      &  0.47$\times$               & 2358.52$\times$     \\
        \hline
    \end{tabular}
    \caption{\ac{DSE} analysis for SDF$^3$ dataflow graph benchmarks}
    \label{tbl:dse-comp}
\end{table*}

We observe that our proposed \gls{hnf} strategy is consistently faster than the $P$ Sweep strategy, achieving a significant speedup of up to $2300\times$.
Moreover, larger speedups are noticed for networks that have a big span between $P_\mathrm{min}$ and $P_\mathrm{max}$.
Impressive speedups are also observed over the $\DecisionVector$ Sweep strategy, except for the MP3 playback network because of the quite small search space of $2^{4}=16$ configurations.
Note that the decision space $\DecisionVector$ has only $\SetCard{\SetActors_\mathrm{SDF}}$ variables, and thus, the decision whether an actor is always-active or self-powered is taken only for each \ac{SDF} actor and not individually for all instances of an actor in the unfolded marked graph.

For the case of the Satellite network, $\DecisionVector$ Sweep failed to explore the design space of $2^{22}$ configurations within 48 hours.
We notice that the $P$ Sweep is faster than $\DecisionVector$ Sweep for large actor networks, whereas for networks with few actors in the \ac{SDFG} ($|\SetActors_{\textrm{SDFG}}| \le 6$), the $\DecisionVector$ Sweep is faster.
\par
A significant observation is that the non-dominated fronts obtained by the \gls{hnf} and $P$ Sweep strategies coincide with the true Pareto fronts $\ParetoRef$ obtained by the $\DecisionVector$ Sweep strategy.
Thus, for all SDF$^3$ benchmark applications, \gls{hnf} finds all Pareto-optimal solutions ($\ParetoFront_\mathrm{\gls{hnf}} = \ParetoRef$).
\par

For the $100$ random graphs, the non-dominated sets returned by our \gls{hnf} strategy match the Pareto fronts obtained by the $\DecisionVector$ Sweep strategy in $98$ out of $100$ cases for $\varepsilon = 1$, while the non-dominated sets obtained by the $P$ Sweep strategy match the true Pareto fronts still in 95 out of 100 cases.
For the 2 cases where $\ParetoFront_\mathrm{\gls{hnf}} \neq \ParetoRef$ ($HV$ ratio $< 1$), we realize that when decrementing the period by $\varepsilon=1$ between iterations (\cref{alg:dse:loop-decrement-period} in \cref{alg:dse}) we might miss to evaluate optimal design points \revised{with a period $P$ lying in this}.
However, when running the same $100$ experiments with a smaller decrement ($\varepsilon=0.1$) already leads to $\ParetoFront_\mathrm{\gls{hnf}} = \ParetoRef$ in all 100 cases.
We identified that for the $5$ cases where $\ParetoFront_{P\textrm{ Sweep}} \neq \ParetoRef$, the $P$ Sweep strategy was obviously not able to find Pareto points that have rational periods.
Notably, \gls{hnf} achieves an average exploration speedup of $29.95\times$ compared to the $P$ Sweep strategy and $342.39\times$ compared to the $\DecisionVector$ Sweep strategy for the 100 graphs tested.
Further facts and figures for the different \ac{DSE} approaches on the 100 graphs are provided in~\cref{tbl:appendix1}.



\section{Related Work}\label{sec:relwork}

In \cite{Jha01LPScheduling, Ahmad15GreenComputing}, two different power management strategies are identified that impact scheduling: \ac{DPM} and \ac{DVS}.
\ac{DPM} involves powering down unused modules to save energy, while in \ac{DVS}, the voltages and frequencies of modules are reduced to fully utilize the idle periods.
The reduction in voltages and/or frequencies leads to power and, eventually, energy savings.
A detailed overview of \ac{DPM} is given in~\cite{Benini00DPM}, while \cite{10.1145/1529255.1529260} discusses \ac{DVS} for dataflow-based \acp{SoC}.
Here, the frequencies of individual functional units are varied, subject to throughput constraints, to optimize for energy.
In \cite{Lu00LPTaskScheduling}, tasks are scheduled to increase the length of the idle period and reduce the number of transitions to and from low-power states to increase energy, using \ac{DPM}.
Moreover, \ac{DPM} integrates well with the notion of dataflow, as is shown by the concept of \emph{self-powering dataflow networks} introduced in~\cite{karim2024-selfpoweringDFGs}.
For this concept, various ways to fully power down an actor and retend its state are realizable, e.g., utilizing non-volatile memories~\cite{karim2025-selfpoweringDFGs-mbmv} or even exploiting upcoming FeMFET-based bit cells for fast and local retention of state~\cite{darne2026-selfpoweringDFGs-FeMFET-state-retention}.
Nonetheless, shutdown and wake-up delays may degrade the achievable throughput, necessitating corresponding analysis techniques.
\par
Extensive studies on the throughput analysis of \acp{DFG} exist in the literature.
The maximum cycle mean~\cite{Fet76,Parhi91OptUnfold} forms the foundation to analyze the throughput of marked graphs, with further extensions explored in \cite{Dasdab98FasterMCM} and \cite{COCHETTERRASSON1998667}.
Apart from marked graphs, there also exist the \ac{SDF}~\cite{E.Lee87SDF} and \ac{CSDF}~\cite{belp_1996-csdf} models of computation, which allow a more concise representation of an application as a \ac{DFG}, but are of equal computational expressiveness~\cite{sgtb_2011-sadf-and-fsm-sadf} as marked graphs because \ac{CSDF} and \ac{SDF} graphs can be translated into marked graphs~\cite{E.Lee86ACH}.
In~\cite{Ghamrain06ThroughputAnalysisSDF}, a well-rounded overview of the throughput analysis of \ac{SDF} graphs is provided, while \cite{BodinK21}, respectively, \cite{KohB22} provide scheduling algorithms for \ac{SDF} and \ac{CSDF} graphs that do not require a translation into marked graphs and, hence, are more computationally efficient than maximum cycle mean-based analysis.
\par
None of the previous works is known to us that treats the problem of performance and power tradeoff analysis for data flow specifications in which actors may have varying execution times due to the fact that they might power down in times of idleness due to either clock- or power-gating induced delays.

\section{Conclusion}\label{sec:conclusion}

The introduction of \acf{DPM} strategies such as clock- and power-gating in dataflow networks has been shown to provide significant energy savings when applied during idle periods.
However, these strategies may degrade achievable throughput due to shutdown and wake-up delays.
Such throughput degradations might be particularly detrimental to signal processing systems that need to sustain a given throughput.
\par
This paper first addresses the fundamental question of whether there exist periodic schedules for dataflow graphs in which, due to data-arrival time dependent activations as a result of shutdown and wake-up delays, the actor execution times may vary.
\par
As a result, we show that if for each actor in such a dataflow graph, the decision is taken that an actor either never goes to sleep (\emph{always-active)} or always runs in a mode called \emph{self-powered}, there exists indeed periodic schedules.
This is shown by  contributing a linear-program formulation for finding a maximal-throughput schedule for a given combination of actors that are exclusively executing in either always-active or self-powered mode in each period (configuration).
\par
Moreover, depending on which configuration is chosen, tradeoffs between throughput and energy savings can be analyzed. This paper presents three strategies for \acf{DSE} of such tradeoffs, including a \ac{DSE} that explores all configurations which is able to determine all Pareto-optimal configurations, an algorithm that avoids to explore all configurations, but finds for a given period $P$ a configuration that has minimal energy $E$, and finally a very efficient strategy called \emph{Hop \& Skip} (H\&S) that is shown to determine also the Pareto front but skipping to search in period intervals in which no Pareto point may exist.
The exploration times as well as the non-dominated sets are compared by the hypervolume.
\gls{hnf} outperforms the other two strategies by speedups of up to three orders of magnitude for real-life and more than $100$ generated benchmark applications of different complexity.

\bibliographystyle{IEEEtran}
\bibliography{literature}

\clearpage
\begin{sidewaystable*}[t]
    \vspace{2mm}
    \centering
    \small
    \setlength{\tabcolsep}{2.5pt}

    \begin{minipage}[h]{9cm}
        \centering
        \begin{tabular}{|c|c|c|c|c|c|c|c|c|c|}
        \hline
        \multirow{2}{*}{$\graph$} & \multirow{2}{*}{$|\SetActors_{\textrm{SDF}}|$} & \multirow{2}{*}{$|\SetActors|$} & \multicolumn{3}{c|}{Exploration Time [s]} & \multicolumn{2}{c|}{Speedup} & \multicolumn{2}{c|}{HV ratio vs XS} \\ 
        \cline{4-10} & & & $\DecisionVector$ Sweep & $P$ Sweep &\gls{hnf} &
        vs $\DecisionVector$ Sweep & vs $P$ Sweep &
        PS/XS & HNF/XS \\ \hline
        1 & 12 & 247 & 109.30 & 55.04 & 1.91 & 57.37$\times$ & 28.89$\times$ & 1.0000 & 1.0000 \\
        2 & 9 & 81 & 2.85 & 15.15 & 0.52 & 5.48$\times$ & 29.09$\times$ & 1.0000 & 1.0000 \\
        3 & 14 & 149 & 468.28 & 47.24 & 3.69 & 126.82$\times$ & 12.79$\times$ & 1.0000 & 1.0000 \\
        4 & 15 & 250 & 890.82 & 41.37 & 2.23 & 399.95$\times$ & 18.57$\times$ & 1.0000 & 1.0000 \\
        5 & 15 & 250 & 3617.65 & 76.22 & 14.24 & 254.05$\times$ & 5.35$\times$ & 1.0000 & 1.0000 \\
        6 & 14 & 249 & 334.22 & 30.62 & 1.75 & 191.40$\times$ & 17.53$\times$ & 1.0000 & 1.0000 \\
        7 & 15 & 250 & 2125.64 & 76.88 & 6.31 & 336.61$\times$ & 12.17$\times$ & 1.0000 & 1.0000 \\
        8 & 13 & 180 & 222.24 & 25.75 & 1.54 & 143.97$\times$ & 16.68$\times$ & 1.0000 & 1.0000 \\
        9 & 14 & 249 & 3084.62 & 127.44 & 10.01 & 308.25$\times$ & 12.74$\times$ & 1.0000 & 1.0000 \\
        10 & 13 & 238 & 126.68 & 23.35 & 1.64 & 77.41$\times$ & 14.27$\times$ & 1.0000 & 1.0000 \\
        11 & 13 & 246 & 233.47 & 49.55 & 3.34 & 69.83$\times$ & 14.82$\times$ & 1.0000 & 1.0000 \\
        12 & 15 & 250 & 776.91 & 41.76 & 3.93 & 197.78$\times$ & 10.63$\times$ & 1.0000 & 1.0000 \\
        13 & 15 & 250 & 2027.11 & 97.63 & 8.49 & 238.87$\times$ & 11.50$\times$ & 1.0000 & 1.0000 \\
        14 & 15 & 250 & 1614.64 & 66.84 & 5.78 & 279.50$\times$ & 11.57$\times$ & 1.0000 & 1.0000 \\
        15 & 15 & 250 & 1256.19 & 57.79 & 4.71 & 266.64$\times$ & 12.27$\times$ & 1.0000 & 1.0000 \\
        16 & 14 & 249 & 861.60 & 48.65 & 4.42 & 195.00$\times$ & 11.01$\times$ & 1.0000 & 1.0000 \\
        17 & 15 & 250 & 1941.18 & 116.59 & 6.37 & 304.76$\times$ & 18.30$\times$ & 1.0000 & 1.0000 \\
        18 & 15 & 250 & 5702.52 & 133.32 & 7.60 & 750.39$\times$ & 17.54$\times$ & 1.0000 & 1.0000 \\
        19 & 15 & 250 & 2660.71 & 137.85 & 5.87 & 453.02$\times$ & 23.47$\times$ & 1.0000 & 1.0000 \\
        20 & 14 & 236 & 547.26 & 81.64 & 3.06 & 178.85$\times$ & 26.68$\times$ & 1.0000 & 1.0000 \\
        21 & 15 & 250 & 957.37 & 77.77 & 8.69 & 110.14$\times$ & 8.95$\times$ & 1.0000 & 1.0000 \\
        22 & 15 & 250 & 2505.03 & 69.47 & 5.10 & 490.92$\times$ & 13.61$\times$ & 1.0000 & 1.0000 \\
        23 & 15 & 250 & 3092.92 & 115.89 & 8.75 & 353.32$\times$ & 13.24$\times$ & 1.0000 & 1.0000 \\
        24 & 11 & 245 & 419.36 & 140.89 & 14.45 & 29.01$\times$ & 9.75$\times$ & 1.0000 & 1.0000 \\
        25 & 15 & 250 & 1170.88 & 97.69 & 3.83 & 305.62$\times$ & 25.50$\times$ & 1.0000 & 1.0000 \\
        26 & 15 & 250 & 863.77 & 40.66 & 4.19 & 206.36$\times$ & 9.71$\times$ & 1.0000 & 1.0000 \\
        27 & 15 & 250 & 5092.72 & 132.72 & 26.02 & 195.70$\times$ & 5.10$\times$ &\textbf{0.6632} & 1.0000 \\
        28 & 15 & 250 & 1628.00 & 47.17 & 2.94 & 552.83$\times$ & 16.02$\times$ & 1.0000 & 1.0000 \\
        29 & 15 & 250 & 927.36 & 60.23 & 3.35 & 276.78$\times$ & 17.97$\times$ & 1.0000 & 1.0000 \\
        30 & 12 & 190 & 112.77 & 22.16 & 1.88 & 60.08$\times$ & 11.81$\times$ & 1.0000 & 1.0000 \\
        31 & 15 & 250 & 873.79 & 59.71 & 3.43 & 254.43$\times$ & 17.39$\times$ & 1.0000 & 1.0000 \\
        32 & 14 & 161 & 616.49 & 38.03 & 3.92 & 157.43$\times$ &  9.71$\times$ &\textbf{0.9960} & \textbf{0.9972} \\
        33 & 15 & 250 & 635.82 & 43.80 & 2.47 & 257.08$\times$ & 17.71$\times$ & 1.0000 & 1.0000 \\
        34 & 15 & 250 & 13895.15 & 1571.86 & 48.12 & 288.75$\times$ & 32.66$\times$ & 1.0000 & 1.0000 \\
        35 & 15 & 250 & 4585.38 & 329.73 & 13.31 & 344.42$\times$ & 24.77$\times$ & 1.0000 & 1.0000 \\
        36 & 14 & 222 & 5148.99 & 703.62 & 51.25 & 100.47$\times$ & 13.73$\times$ & 1.0000 & 1.0000 \\
        37 & 12 & 189 & 803.50 & 271.31 & 19.21 & 41.84$\times$ & 14.13$\times$ & 1.0000 & 1.0000 \\
        38 & 15 & 250 & 8173.09 & 435.40 & 11.08 & 737.79$\times$ & 39.30$\times$ & 1.0000 & 1.0000 \\
        39 & 15 & 250 & 4994.74 & 274.40 & 30.00 & 166.48$\times$ & 9.15$\times$ & 1.0000 & 1.0000 \\
        40 & 15 & 250 & 8262.25 & 750.85 & 14.24 & 580.18$\times$ & 52.73$\times$ & 1.0000 & 1.0000 \\
        41 & 15 & 250 & 7457.87 & 575.48 & 19.64 & 379.66$\times$ & 29.30$\times$ & 1.0000 & 1.0000 \\
        42 & 15 & 250 & 12244.84 & 416.23 & 21.28 & 575.41$\times$ & 19.56$\times$ & \textbf{0.9810} & \textbf{0.9810} \\
        43 & 14 & 249 & 6141.19 & 749.03 & 25.92 & 236.96$\times$ & 28.90$\times$ & 1.0000 & 1.0000 \\
        44 & 15 & 250 & 2787.38 & 234.07 & 5.36 & 520.03$\times$ & 43.67$\times$ & 1.0000 & 1.0000 \\
        45 & 15 & 250 & 5375.15 & 323.12 & 16.36 & 328.50$\times$ & 19.75$\times$ & 1.0000 & 1.0000 \\
        46 & 15 & 250 & 7638.80 & 482.55 & 9.77 & 781.67$\times$ & 49.38$\times$ & 1.0000 & 1.0000 \\
        47 & 15 & 250 & 3976.34 & 98.41 & 7.73 & 514.28$\times$ & 12.73$\times$ & 1.0000 & 1.0000 \\
        48 & 15 & 250 & 10081.27 & 518.74 & 11.49 & 877.30$\times$ & 45.14$\times$ & 1.0000 & 1.0000 \\
        49 & 15 & 250 & 9779.81 & 680.07 & 10.91 & 896.00$\times$ & 62.31$\times$ & 1.0000 & 1.0000 \\
        50 & 15 & 250 & 3927.05 & 157.04 & 16.47 & 238.37$\times$ & 9.53$\times$ & 1.0000 & 1.0000 \\
        \hline
        \end{tabular}
    \end{minipage}\hfill
    \begin{minipage}[h]{0.49\textheight}
        \centering
        \begin{tabular}{|c|c|c|c|c|c|c|c|c|c|}
        \hline
        \multirow{2}{*}{$\graph$} & \multirow{2}{*}{$|\SetActors_{\textrm{SDF}}|$} & \multirow{2}{*}{$|\SetActors|$} &
        \multicolumn{3}{c|}{Exploration Time [s]} &
        \multicolumn{2}{c|}{Speedup} &
        \multicolumn{2}{c|}{HV ratio vs XS} \\
        \cline{4-10} & & &
        $\DecisionVector$ Sweep & $P$ Sweep & \gls{hnf} &
        vs $\DecisionVector$ Sweep & vs $P$ Sweep &
        PS/XS & HNF/XS \\ \hline
        51 & 15 & 250 & 10082.38 & 725.37 & 22.02 & 457.84$\times$ & 32.94$\times$ & 1.0000 & 1.0000 \\
        52 & 15 & 250 & 6480.28 & 368.92 & 15.08 & 429.85$\times$ & 24.47$\times$ & 1.0000 & 1.0000 \\
        53 & 14 & 249 & 5735.62 & 554.96 & 28.86 & 198.74$\times$ & 19.23$\times$ & 1.0000 & 1.0000 \\
        54 & 15 & 250 & 12753.47 & 1272.70 & 23.15 & 551.01$\times$ & 54.99$\times$ & 1.0000 & 1.0000 \\
        55 & 14 & 249 & 2124.77 & 192.41 & 9.24 & 229.92$\times$ & 20.82$\times$ & 1.0000 & 1.0000 \\
        56 & 15 & 250 & 12361.66 & 1200.01 & 21.05 & 587.20$\times$ & 57.00$\times$ & 1.0000 & 1.0000 \\
        57 & 14 & 249 & 9235.70 & 1251.33 & 27.89 & 331.14$\times$ & 44.87$\times$ & 1.0000 & 1.0000 \\
        58 & 13 & 140 & 1081.66 & 197.33 & 15.37 & 70.37$\times$ & 12.84$\times$ & 1.0000 & 1.0000 \\
        59 & 14 & 236 & 8159.55 & 430.24 & 19.01 & 429.33$\times$ & 22.64$\times$ & 1.0000 & 1.0000 \\
        60 & 14 & 244 & 3197.37 & 253.62 & 16.45 & 194.39$\times$ & 15.42$\times$ & 1.0000 & 1.0000 \\
        61 & 15 & 250 & 13894.37 & 969.59 & 26.92  & 516.14$\times$ & 36.02$\times$ & 1.0000 & 1.0000 \\
        62 & 14 & 249 & 3470.66 & 503.46 & 33.94 & 102.26$\times$ & 14.83$\times$ & 1.0000 & 1.0000 \\
        63 & 15 & 250 & 5767.06 & 334.33 & 22.44 & 256.98$\times$ & 14.90$\times$ & 1.0000 & 1.0000 \\
        64 & 15 & 250 & 4987.55 & 247.16 & 21.02 & 237.30$\times$ & 11.76$\times$ & 1.0000 & 1.0000 \\
        65 & 14 & 245 & 2516.95 & 342.46 & 14.61 & 172.31$\times$ & 23.45$\times$ & 1.0000 & 1.0000 \\
        66 & 15 & 250 & 8489.80 & 620.04 & 22.07 & 384.76$\times$ & 28.10$\times$ & 1.0000 & 1.0000 \\
        67 & 14 & 249 & 5888.91 & 765.66 & 64.04 & 91.96$\times$ & 11.96$\times$ & 1.0000 & 1.0000 \\
        68 & 12 & 12 & 52.34 & 37.24 & 1.57 & 33.36$\times$ & 23.73$\times$ & 1.0000 & 1.0000 \\
        69 & 15 & 250 & 5730.75 & 265.22 & 11.25 & 509.47$\times$ & 23.58$\times$ & 1.0000 & 1.0000 \\
        70 & 15 & 250 & 10716.24 & 928.26 & 15.56 & 688.69$\times$ & 59.66$\times$ & 1.0000 & 1.0000 \\
        71 & 15 & 250 & 4153.86 & 144.64 & 7.22 & 575.20$\times$ & 20.03$\times$ & 1.0000 & 1.0000 \\
        72 & 14 & 223 & 2779.05 & 239.79 & 15.02 & 185.07$\times$ & 15.97$\times$ & 1.0000 & 1.0000 \\
        73 & 14 & 237 & 2579.97 & 320.48 & 15.79 & 163.43$\times$ & 20.30$\times$ & 1.0000 & 1.0000 \\
        74 & 15 & 250 & 7680.95 & 349.71 & 10.93 & 702.47$\times$ & 31.98$\times$ & 1.0000 & 1.0000 \\
        75 & 15 & 250 & 5870.87 & 543.69 & 28.83 & 203.67$\times$ & 18.86$\times$ & 1.0000 & 1.0000 \\
        76 & 14 & 244 & 11226.92 & 1619.76 & 37.42 & 300.05$\times$ & 43.29$\times$ & 1.0000 & 1.0000 \\
        77 & 15 & 250 & 4549.64 & 264.83 & 15.32 & 296.99$\times$ & 17.29$\times$ & \textbf{0.9798} & 1.0000 \\
        78 & 13 & 244 & 6502.73 & 1451.84 & 38.15 & 170.45$\times$ & 38.06$\times$ & 1.0000 & 1.0000 \\
        79 & 15 & 250 & 7687.53 & 543.28 & 23.66 & 324.95$\times$ & 22.96$\times$ & 1.0000 & 1.0000 \\
        80 & 15 & 250 & 12934.24 & 1261.33 & 16.03 & 806.68$\times$ & 78.67$\times$ & 1.0000 & 1.0000 \\
        81 & 15 & 250 & 5595.85 & 239.72 & 9.43 & 593.66$\times$ & 25.43$\times$ & 1.0000 & 1.0000 \\
        82 & 14 & 200 & 3211.65 & 293.39 & 15.08 & 212.92$\times$ & 19.45$\times$ & 1.0000 & 1.0000 \\
        83 & 15 & 250 & 4318.12 & 228.29 & 11.35 & 380.54$\times$ & 20.12$\times$ & 1.0000 & 1.0000 \\
        84 & 14 & 249 & 2382.70 & 257.88 & 15.06 & 158.22$\times$ & 17.12$\times$ & 1.0000 & 1.0000 \\
        85 & 15 & 250 & 8255.24 & 467.37 & 13.11 & 629.90$\times$ & 35.66$\times$ & 1.0000 & 1.0000 \\
        86 & 15 & 250 & 14199.51 & 916.18 & 14.85 & 956.30$\times$ & 61.70$\times$ & 1.0000 & 1.0000 \\
        87 & 15 & 250 & 8540.53 & 245.32 & 13.49 & 632.95$\times$ & 18.18$\times$ & 1.0000 & 1.0000 \\
        88 & 8 & 119 & 11.07 & 76.67 & 3.96 & 2.79$\times$ & 19.36$\times$ & 1.0000 & 1.0000 \\
        89 & 15 & 250 & 8907.12 & 672.21 & 15.13 & 588.85$\times$ & 44.44$\times$ & 1.0000 & 1.0000 \\
        90 & 15 & 250 & 11812.49 & 935.24 & 17.35 & 681.01$\times$ & 53.92$\times$ & 1.0000 & 1.0000 \\
        91 & 15 & 250 & 6363.85 & 242.23 & 23.05 & 276.14$\times$ & 10.51$\times$ & \textbf{0.9992} & 1.0000 \\
        92 & 14 & 249 & 3062.48 & 322.49 & 21.09 & 145.21$\times$ & 15.29$\times$ & 1.0000 & 1.0000 \\
        93 & 15 & 250 & 3351.61 & 91.93 & 11.20 & 299.13$\times$ & 8.20$\times$ & 1.0000 & 1.0000 \\
        94 & 14 & 249 & 3156.01 & 379.96 & 18.10 & 174.35$\times$ & 20.99$\times$ & 1.0000 & 1.0000 \\
        95 & 14 & 167 & 1165.72 & 187.72 & 7.69 & 151.50$\times$ & 24.40$\times$ & 1.0000 & 1.0000 \\
        96 & 15 & 250 & 9583.05 & 704.82 & 32.31 & 296.57$\times$ & 21.81$\times$ & 1.0000 & 1.0000 \\
        97 & 15 & 250 & 9617.25 & 420.46 & 15.89 & 605.23$\times$ & 26.46$\times$ & 1.0000 & 1.0000 \\
        98 & 15 & 250 & 4577.22 & 245.29 & 9.05 & 505.84$\times$ & 27.11$\times$ & 1.0000 & 1.0000 \\
        99 & 13 & 178 & 714.50 & 89.21 & 6.24 & 114.42$\times$ & 14.29$\times$ & 1.0000 & 1.0000 \\
        100 & 15 & 250 & 10413.67 & 459.30 & 11.11 & 937.48$\times$ & 41.35$\times$ & 1.0000 & 1.0000 \\
        \hline
        \end{tabular}
    \end{minipage}

    \caption{\ac{DSE} analysis for 100 randomly generated dataflow graphs}
    \label{tbl:appendix1}
\end{sidewaystable*}

\end{document}